\documentclass{article}

\usepackage[preprint, nonatbib]{neurips_2024}

\usepackage[utf8]{inputenc} 
\usepackage[T1]{fontenc}    
\usepackage{lmodern}  

\usepackage{hyperref}       
\usepackage{url}            
\usepackage{booktabs}       
\usepackage{amsfonts}       
\usepackage{nicefrac}       
\usepackage{microtype}      

\usepackage[backend=biber,
maxcitenames = 2,
doi=false,
isbn=false,
url=false,
uniquename=false,
maxbibnames=3,
style=bwl-FU]{biblatex}
\DeclareFieldFormat[article]{title}{\mkbibitalic{#1}}
\DeclareFieldFormat{journaltitle}{#1\isdot}

\DeclareFieldFormat[inproceedings]{title}{\mkbibitalic{#1}}
\DeclareFieldFormat[inproceedings]{booktitle}{#1}

\renewbibmacro*{volume+number+eid}{%
	,\setunit*{\addnbspace}
	Vol. \printfield{volume},%
	\iffieldundef{number}{}{
		No. \printfield{number},}
	\printfield{eid}}
\DeclareFieldFormat[article]{number}{#1}
\bibliography{biblio.bib}

\usepackage{amsmath,amssymb,amsthm,amsfonts}
\usepackage{mathtools}  
\usepackage{mathrsfs}
\usepackage{caption}
\usepackage{subcaption}
\usepackage{color}
\usepackage{float}
\usepackage{array}
\usepackage{comment}

\usepackage{graphicx}

\usepackage{enumitem}

\usepackage[normalem]{ulem} 

\newcolumntype{L}{>{$}l<{$}} 
\newcolumntype{R}{>{$}r<{$}} 
\newcolumntype{C}{>{$}c<{$}} 
\usepackage{multirow}
\usepackage{booktabs} 
\newcolumntype{x}[1]{>{\centering\arraybackslash\hspace{0pt}}p{#1}} 
\usepackage{bm}
\usepackage[symbol]{footmisc}
\usepackage{tabularx}
\usepackage[table]{xcolor}
\usepackage{color, colortbl}

\usepackage[ruled]{algorithm2e}
\SetKwInput{KwParam}{Parameters}

\usepackage{cleveref}  

\newtheorem{proposition}{Proposition}

\newtheorem{lemma}{Lemma}

\definecolor{mygreen}{RGB}{0,142,0}
\definecolor{shadowed}{gray}{0.9}

\newcommand{\blue}[1]{\textcolor{blue}{#1}}

\newcommand{\red}[1]{\textcolor{red}{#1}}
\newcommand{\green}[1]{\textcolor{mygreen}{#1}}

\usepackage{bm}

\DeclareRobustCommand{\vect}[1]{\bm{#1}}

\DeclareMathOperator{\normal}{Normal}
\DeclareMathOperator{\bernoulli}{Bernoulli}
\DeclareMathOperator{\multinomial}{Multinomial}
\DeclareMathOperator{\dirichlet}{Dirichlet}

\DeclareMathOperator{\KL}{KL}

\DeclareMathOperator{\Tr}{Tr}

\DeclareMathOperator{\softmax}{softmax}

\DeclareMathOperator{\enc}{Encoder}
\DeclareMathOperator{\dec}{Decoder}

\newcommand{\bfz}{\mathbf{z}}

\newcommand{\bfA}{\mathbf{A}}
\newcommand{\bfD}{\mathbf{D}}

\newcommand{\bfZ}{\mathbf{Z}}
\newcommand{\bfW}{\mathbf{W}}

\newcommand{\bfP}{\mathbf{P}}

\newcommand{\bfL}{\mathbf{L}}
\newcommand{\bfI}{\mathbf{I}}
\newcommand{\bfU}{\mathbf{U}}

\newcommand{\bfH}{\mathbf{H}}
\newcommand{\bfeta}{\bm{\eta}}

\newcommand{\bfPi}{\bm{\Pi}}

\newcommand{\calm}{\mathcal{M}}
\newcommand{\caln}{\mathcal{N}}

\newcommand{\bbr}{\mathbb{R}}

\DeclareFontFamily{U}{zapfc}{\skewchar \font =45}
\DeclareFontShape{U}{zapfc}{m}{n}{<-> pzcmi}{}
\DeclareMathAlphabet{\mathcalzapf}{U}{zapfc}{m}{n}
\newcommand{\zapfm}{\mathcalzapf{M}}	

\title{The Deep Latent Position Block Model For The Block Clustering And Latent Representation Of Networks}

\author{%
  Rémi~Boutin\\
  Laboratoire de Probabilités, Statistique et Modélisation, CNRS, UMR 8001\\
  Sorbonne Université\\
  Paris, France\\
  \texttt{remi.boutin.stat@gmail.com} \\
  \AND
  Pierre~Latouche\\
  Laboratoire de Mathématiques Blaise Pascal, CNRS, UMR 6620\\
  Université Clermont Auvergne\\
  Aubière, France\\
  \AND
  Charles~Bouveyron\\
  Laboratoire J-A. Dieudonné, CNRS, UMR 7351, INRIA, MAASAI Team\\
  Université Côte d'Azur\\
  Sophia-Antipolis, France\\
}

\begin{document}

\maketitle

\begin{abstract}
	The increased quantity of data has led to a soaring use of networks to model relationships between different objects, represented as nodes. Since the number of nodes can be particularly large, the network information must be summarised through node clustering methods. In order to make the results interpretable, a relevant visualisation of the network is also required. To tackle both issues, we propose a new methodology called deep latent position block model (Deep LPBM) which simultaneously provides a network visualisation coherent with block modelling, allowing a clustering more general than community detection methods, as well as a continuous representation of nodes in a latent space given by partial membership vectors. Our methodology is based on a variational autoencoder strategy, relying on a graph convolutional network, with a specifically designed decoder. The inference is done using both variational and stochastic approximations. In order to efficiently select the number of clusters, we provide a comparison of three model selection criteria. An extensive benchmark as well as an evaluation of the partial memberships are provided. We conclude with an analysis of a French political blogosphere network and a comparison with another methodology to illustrate the insights provided by Deep LPBM results.
\end{abstract}

\section{Introduction and motivation}
Graph-structured data are ubiquitous in many scientific fields such as in social sciences or in Biology. They are able to represent any type of interaction between any kind of objects. With the network sizes increasing, as well as their complexity, it is necessary to develop techniques rendering comprehensive information about their latent structure. Therefore, the aim of this work is twofold. The first necessary step is to estimate node partial memberships to apprehend the connectivity patterns in the network. The second critical step towards having a grasp on the content of a graph is to obtain a meaningful visualisation of the dataset.

\newlength{\HEIGHT}
\newlength{\lowheight}
\begin{figure}
	\centering
	\settoheight{\lowheight}{\vbox{
			\subcaptionbox{$\bfPi$}[0.33\linewidth]{
				$
				\begin{pmatrix}
					0.01 & 0.3 & 0.3 \\
					0.3 & 0.01 & 0.3 \\
					0.3 & 0.3 & 0.01
				\end{pmatrix}
				$
	}}}
	\settoheight{\HEIGHT}{\vbox{
			\subcaptionbox{Deep LPBM}[0.28\linewidth]{
				\includegraphics[width=\linewidth]
				{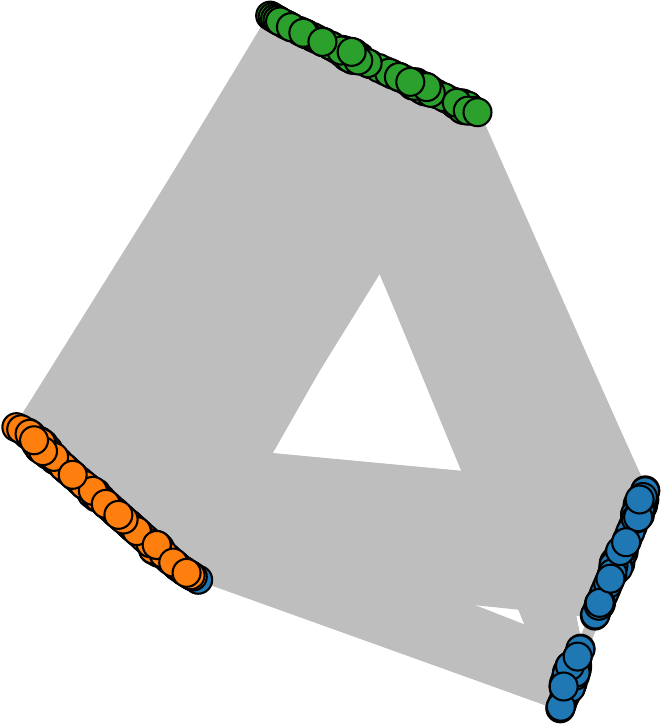}
	}}}
	\subcaptionbox{$\bfPi$}[0.33\linewidth]{
		$
		\begin{pmatrix}
			0.01 & 0.3 & 0.3 \\
			0.3 & 0.01 & 0.3 \\
			0.3 & 0.3 & 0.01
		\end{pmatrix}
		$
		\vspace*{0.35 \HEIGHT - 0.35 \lowheight}
	}
	\hfill
	\subcaptionbox{FR}[0.28\linewidth]{
		\includegraphics[width=\linewidth]
		{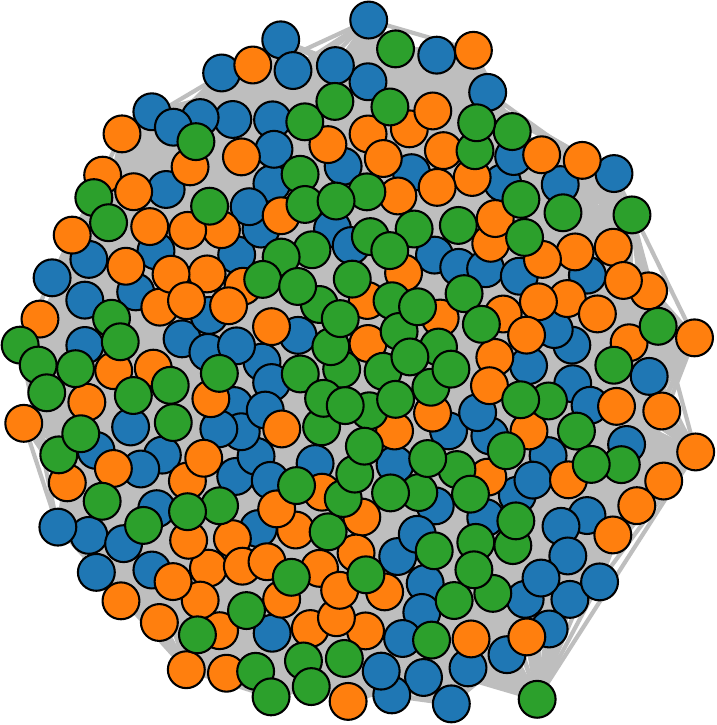}
	}
	\hfill
	\subcaptionbox{Deep LPBM}[0.28\linewidth]{
		\includegraphics[width=\linewidth]
		{cropped-deeplpbm_disassortative_example}
	}
	\caption{For the same disassortative network, Fruchterman-Reingold node layout (FR) provides a uninformative visualisation while Deep LPBM node layout nicely renders the true connectivity patterns induced by the connectivity matrix $\bfPi$, corresponding to the probability of connection between clusters. The node colour corresponds to their corresponding cluster and the probability of connection between clusters are given in the Figure on the left-hand side.}
	\label{fig:example_node_layout_introduction}
\end{figure}

\paragraph{Main contributions}

In this paper, we address several shortcomings of canonical block models, as well as positional methods, by proposing a novel graph variational autoencoder named deep latent position block model (Deep LPBM). This methodology is developed in an unsupervised framework, and no node label is required. Therefore, Deep LPBM focuses on capturing patterns responsible for the observed data and should not be considered a node classification methodology. We summarise the main contributions of this paper below:
\begin{itemize}
	\item  We propose a novel block-structured decoder, called Deep LPBM (deep latent position block model), combined with a graph convolutional network (GCN) based encoder, to model any type of connectivity pattern.
	\item Deep LPBM generalises many random graph models for network analysis
	\item By using partial memberships, Deep LPBM is able to associate each node with several connectivity patterns, rendering refined results as illustrated in the analysis of the French political blogosphere.
	\item To the best of our knowledge, Deep LPBM is the first method capable of simultaneously \textit{i)} using a variational graph autoencoder algorithm, \textit{ii)} providing a visualisation of the entire network compatible with block modelling and \textit{iii)} performing block modelling, as well as node partial membership estimation.
\end{itemize}

\begin{figure}
	\centering
	\includegraphics[width=0.3 \linewidth]{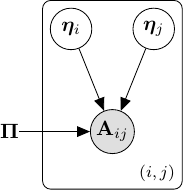}
	\caption{Graphical model associated to Deep LPBM. $\bfA_{ij}$ corresponds to the observed data (in a shadowed circle), $\bfeta_i$ and $\bfeta_i$ to the node latent representations (in an empty circle) and $\bfPi$ to the parameters (not in a circle).}
	\label{fig:graphical_model}
\end{figure}

\section{Model}
\label{sec:model}
This section presents the Deep LPBM modelling assumptions, as well as the links with other random graph models for networks. 

\paragraph{Notations and data}
We start by describing the data considered in this work and  the notations used in this paper.

First, the present methodology is interested in undirected graph-structured data, denoted $\mathcal{G} = \{\mathcal{V}, \mathcal{E}\}$, where $ \mathcal{V} = \{1, \dots, N\}$ corresponds to the set of vertices with cardinal $N$, and $\mathcal{E}=\{ (i,j)\in \mathcal{V}^2: \text{$i$ and $j$ are connected} \}$ to the set edges with cardinal $M$. The adjacency matrix $\bfA = (\bfA_{ij})_{1\leq i,j \leq N}$ is a $N \times N$ binary matrix such that $\bfA_{ij} = 1$ if $i$ and $j$ are connected, $0$ otherwise. Since the graph is assumed to be undirected, the adjacency matrix is symmetric. Second, the number of node clusters will be denoted $Q$. The bijective softmax function is a bijective mapping between $\mathbb{R}^{d-1}$ and the $d$-dimensional simplex $\Delta_d$, such that for any $x \in \mathbb{R}^d$ and any $y \in \Delta_d$:
\begin{align}
	&\softmax(x)  = \Bigl(1 + \sum_{k'=1}^{Q-1} e^{x_{k'}}\Bigr)^{-1} ( e^{x_1}, \dots, e^{x_{Q-1}}, 1), \nonumber\\
	&\softmax^{-1}(y)  = \left( \ln\left(\frac{y_1}{y_Q}\right), \dots,   \ln \left( \frac{y_{Q-1}}{y_Q} \right) \right). \label{eq:def_softmax_inv}
\end{align}

\paragraph{Generative model}
First, we assume a vector of probabilities is assigned to each node $i$, modelling the partial memberships of the corresponding node, modelled by a logistic-normal distribution. Setting $d=Q-1$, we consider:
\begin{equation}\label{eq:eta_logisticnorm_prior}
	\begin{split}
		\bfz_{\bm{i}} & \overset{i.i.d}{\sim} \normal_d(0, \bfI_{d}),\\
		\bfeta_{\bm{i}} & \overset{\phantom{i.i.d}}{=} \softmax(\bfz_{\bm{i}}).
	\end{split}	
\end{equation}
This assumption illustrates the possibility for each node to partially belong to multiple clusters. Hence, each vector $\bfeta_{i}^{\top}=( \eta_{i1} \dots \eta_{iQ})$ corresponds to node $i$ partial memberships such that the proportion of node $i$ associated to cluster $q$ is given by $\eta_{iq}$.

Second, given the partial memberships $(\bfeta_i)_{i=1,\dots,N}$, denoted $\bfeta$, the edges are assumed to be independent and to occur according to the following distribution:
\begin{equation}\label{eq:bernouilli}
	\bfA_{ij}	\mid \bfeta_{\bm{i}}, \bfeta_{\bm{j}} \overset{i.i.d}{\sim} \bernoulli( \bfeta_{\bm{i}}^{\top} \bfPi \bfeta_{\bm{j}}).
\end{equation}
The $Q \times Q$ matrix $\bfPi = (\bfPi_{qr})_{1\leq q,r\leq Q}$ is symmetric, with entries $\bfPi_{qr} \in [0, 1]$ corresponding to connection probabilities.  Therefore, an edge between $i$ and $j$ exists with the following probability:
\begin{equation}
	p(A_{ij} = 1 \mid \bfeta_{\bm{i}}, \bfeta_{\bm{j}}, \bfPi)~=~\bfeta_{i}^{\top} \bfPi \bfeta_j.
\end{equation}

\section{Related work and link with Deep LPBM}
This section presents works related to block modelling and positional modelling as well as their links with Deep LPBM.

\subsection{Block modelling}
The stochastic block model \parencite[SBM,][]{wang1987stochastic, snijders1997estimation, daudin2008mixture} considers $\bfeta_i$ as a binary variable $\bfeta_{\bm{i}} \sim \multinomial_{Q}(1; \alpha)$, encoding the cluster of node $i$. Given the clusters of nodes $i$ and $j$ are $q$ and $r$, corresponding to $\eta_{iq} = \eta_{jr} = 1$, the probability of connection in \Cref{eq:bernouilli} would result in:
\begin{equation*}
	p(A_{ji}~=~1 \mid \bfeta_{\bm{i}}, \bfeta_{q}, \bfPi)~=~ \bfeta_{\bm{i}}^{\top} \bfPi \bfeta_{\bm{j}}~=~\Pi_{qr}.
\end{equation*}
Deep LPBM, by using a logistic-normal prior on $\bfeta_{\bm{i}}$, relax the binary constraint on the cluster membership variable $\bfeta_{\bm{i}}$. This type of relaxation has been studied in the context of exponential models by \cite{heller2008statistical}, with factorisable distributions, which is not the case here due to the dependence on both $i$ and $j$ in \Cref{eq:bernouilli}. SBM has been extended to incorporate mixed-membership in  \cite[MMSBM,][]{airoldi2008mixed}. In this model, each node $i$ plays a specific role with respect to a corresponding edge. For instance, considering the edge between $i$ and $j$, the role of each node is modelled by a membership indicator $U_{ij} \sim \multinomial_{Q}(1; \eta_{i} )$ for the role of $i$ and $U_{ji} \sim \multinomial_{Q}(1; \eta_{j})$ for the role of $j$, with the vector of probabilities $\eta_i$ sampled according to $\eta_i \sim \dirichlet_{Q}(\gamma)$. While the authors aimed at considering the quantity $p(\bfA \mid \bm{U}, \bfPi)$, where $\bm{U}$ is the set of all latent vectors $U_{ij}$, the marginal quantity $p(\bfA_{ij} = 1 \mid \bm{\eta}_i, \bm{\eta}_j, \bfPi)$ gives the following probability of connection for any $i < j$:
\begin{equation*}
	\begin{split}
		p(\bfA_{ij} = 1 \mid \bm{\eta}_i, \bm{\eta}_j, \bfPi) = \bfeta_i^{\top} \bfPi \bfeta_j.
	\end{split}
\end{equation*}
We retrieve the same probability of connection as in Equation \eqref{eq:bernouilli}. However, MMSBM assumes a Dirichlet a priori distribution on the partial membership vectors while we suppose a logistic-normal distribution. Additionally, MMSBM relies on a different inference strategy, while, as we shall detail below, we rely on a variational EM algorithm to incorporate graph neural networks and leverage their powerful encoding capacity. In the inference of Deep LPBM, the set $\bm{U}$ is never considered which strongly reduces the number of latent variables to handle and to estimate.

\subsection{Positional models and links with block modelling}
The latent position model \parencite[LPM,][]{hoff2002latent} is considered as the seminal work regarding positional models. LPM assumes that each vertex $i$ is represented by a point in a Euclidean latent space denoted $\bfeta_i \in \mathbb{R}^d$. Given the vertex positions $\bfeta_i$ and $\bfeta_j$ of nodes $i$ and $j$ respectively, an edge between the two exists with probability $f(\bfeta_i, \bfeta_i)$, where $f$ corresponds to a \textit{link function}, also named \textit{kernel function}. For instance, in \cite{hoff2002latent}, the authors considered $f(\bfeta_i, \bfeta_i) = \|\bfeta_i - \bfeta_j \|_2$ the Euclidean distance. Conditional on the node positions, the edges are assumed to be independent. For undirected graphs, and in the absence of covariates, the scalar product is useful in practice since it is fast to compute and it sets nodes sharing close connectivity patterns along similar directions. It was extended in \cite{handcock2007model} by incorporating clustering into the modelling. Lately, developments regarding deep generative models, introduced variational graph autoencoder \parencite[VGAE,][]{kipf2016variational}, focused on encoding the latent structure with a graph convolutional network \parencite[GCN,][]{kipf2016semi}, with a decoder, or link function, based on the dot product $f(\bfeta_i, \bfeta_i) = \bfeta_i^{\top} \bfeta_j$. This constrains the model to respect the transitivity property, \textit{the friend of my friend is my friend} effect \parencite{newman2002assortative}, and limits the connectivity patterns that can be detected. A star pattern, common in social networks, necessitates a more general approach such as block modelling. 

A few works have aimed at bridging the gap between block modelling and positional modelling. For instance, the latent variable model of relational data \parencite{hoff2007modeling} links the probability of connection between two nodes with $\eta_i^{\top} \bfPi \eta_i$ through a probit function. However, contrary to the proposed model, $\eta_i$ is a vector of free parameters in $\bbr^{Q}$ and $\bfPi \in \calm_{Q \times Q}(\bbr)$ is a diagonal matrix with entries that may be positives or negatives. First, Deep LPBM does not assume a specific form of the matrix $\bfPi$ but constrained its values between $0$ and $1$ to ease the interpretation. Second, this model does not assume a generative assumption for each variable nor introduce a node cluster membership variable allowing to model the inter-cluster connectivity. 

To overcome the first limitation raised above, \cite{daudin2010model} proposed the extremal vertices model for random graph (EVMRG). It also relies on the marginalisation of MMSBM, but considers $\bfeta$ as a parameter and not as a random variable, preventing from incorporating the noise within the partial memberships. Moreover, EVRMG inference is based on a linear approximation of the log likelihood, preventing from using an autoencoding framework as well as graph neural network representational power, as we shall detail bellow.

The generalised random dot product graph \parencite[GRDPG,][]{rubin2022statistical} relies on continuous node representation, and is general enough so that it  incorporates block models such as SBM and MMSBM. Under this model, each node $i$, is assigned a vector $X_i$, with conditions on $X_i$ such that for any $i < j$, $\sum_{q=1}^p X_{iq} X_{jq}~-~\sum_{k=p+1}^{p+d} X_{ik} X_{jk} \in [0,1]$. Hence,
\[
\bfA_{ij}~\sim~\operatorname{Bernoulli} \Bigl( \sum_{q=1}^p X_{iq} X_{jq}~-~\sum_{k=p+1}^{p+d} X_{ik} X_{jk} \Bigr).
\]
As noted in \cite{rubin2022statistical}, denoting $p$ and $q$ the number of strictly positive and strictly negative eigenvalues of $\bfPi$, and put $d=p + q$, by choosing $v_1, \dots, v_Q \in \mathbb{R}^d$ such that $v_q^{\top} \mathbf{I}_{p,q} v_{r} = \bfPi_{qr}$, imposing that $X_i = \sum_{q=1}^Q \eta_{iq} v_q$, the probability of existence of an edge becomes $X_i^{\top} \mathbf{I}_{p,q} X_j =  \eta_i^{\top} \bfPi \eta_{j}$. We retrieve Deep LPBM  probability which indicates that Deep LPBM is a special case of the GRDPG. Therefore, insights from GRDPG, notably on the identifiability, are relevant to the Deep LPBM methodology. Nonetheless, the differences in the corresponding generative models and in the inference are key in Deep LPBM to benefit from the efficiency and flexibility of a variational graph autoencoding framework as described in the next section.

\section{Inference}\label{sec:inference_optim}
To estimate the value of the connectivity matrix $\bfPi$, we aim at computing the marginal log-likelihood of the data:
\begin{equation}
	\begin{split}
		\ln p(\bfA \mid \bfPi) = \ln \int_{\bfZ}  p(\bfA, \bfZ \mid \bfPi) d \bfZ,
	\end{split}
\end{equation}
where $\bfZ$ denotes the set all latent vectors $\bfz_i$, for all $i$ in $\{1, \dots, N\}$. Unfortunately, this quantity is not tractable because of the softmax function. In addition, an expectation-maximisation (EM) algorithm cannot be employed directly since the posterior distribution $p(\bfZ \mid \bfA, \bfPi)$ is not tractable. Indeed, $p(\bfZ \mid \bfA, \bfPi) \neq \prod_{i=1}^N p(\bfz_{\bm{i}} \mid \bfA, \bfPi)$ and $p(\bfz_{\bm{i}} \mid \bfA, \bfPi)$ depends on the entire adjacency matrix $\bfA$, preventing from computing this quantity. This problem arises because of the nature of the graphical model, with the two-to-one relationship between latent vectors and observed data (Figure \ref{fig:graphical_model}). It is at the core of all SBM based strategies \parencite{daudin2008mixture}. Consequently, we rely on a variational EM algorithm \parencite{wainwright2008graphical} to tackle this issue.

\paragraph{Variational EM algorithm} 
The variational inference algorithm introduces $R(\bfZ)$, the variational distribution, which serves as a surrogate of the posterior distribution $p(\bfZ \mid A, \bfPi)$. Doing so permits to decompose the marginal log-likelihood for any distribution $R(\bfZ)$:
\begin{equation*} 
	\begin{split}
		\ln p(\bfA \mid \bfPi) = \mathscr{L}(\bfPi; R) + \KL( R(\bfZ) \mid \mid p(\bfZ \mid \bfA, \bfPi)),
	\end{split}
\end{equation*}
with on the left-hand side $\mathscr{L}(\bfPi; R)$, the expected lower bound (ELBO) defined in \Cref{eq:ELBO} and on the right-hand side, the Kullback-Leibler divergence between the two distributions. It is worth noticing that the Kullback-Leibler is always non-negative and thus, the ELBO is a lower bound of the marginal log-likelihood. Moreover, the ELBO is defined as:
\begin{equation}\label{eq:ELBO}
	\begin{split}
		\mathscr{L}(\bfPi; R) & = \mathbb{E}_{R(\bfZ)} \left[ \ln \frac{p(\bfA, \bfZ \mid \bfPi)}{R(\bfZ)}\right].
	\end{split}
\end{equation}
Let us remark that the closer $R(\bfZ)$ is to $p(\bfZ \mid \bfA, \bfPi)$ in terms of Kullback-Leibler divergence, the tighter the lower bound is. To obtain explicit expressions, it is necessary to restrict the family of considered variational distributions by making assumptions described in the following section.

\paragraph{Assumptions on the variational distribution} 
First, we assume that the variational distribution respects the mean-field hypothesis, also referred to as the total factorisation assumption. Second, we assume a specific parametrisation of the distribution to encode the data using a graph neural network such that:
\begin{equation}
	R_{\phi}(\bfZ)~=~\prod_{i=1}^N \caln_{d} (\bfz_i; \mu_{\phi}(\bfA)_i, \sigma_{\phi}(\bfA)_i^2 \bfI_{d}),
\end{equation}
with $\mu_{\phi}(\bfA)$ (respectively $\sigma_{\phi}(\bfA)^2$) corresponding to the variational means (resp. variances) encoded by the canonical VGAE encoder \parencite{kipf2016variational} defined as $\bfH_1 = \operatorname{ReLu}(\tilde{\bfL} \bfW_0)$, $\mu_{\phi}(\bfA) = \tilde{\bfL} \bfH_1  \bfW_{\mu}$ and $\log \sigma_{\phi}(\bfA)  = \tilde{\bfL} \bfH_1  \bfW_{\sigma}$, with $\tilde{\bfA} = \bfA + \bfI_{N}$,  $\tilde{\bfD}$ the diagonal matrix with $\tilde{\bfD}_{ii}~=~\sum_{j=1}^N \tilde{\bfA}_{ij}$ for any node $i$, and $\tilde{\bfL} = \tilde{\bfD}^{-1/2} \tilde{\bfA}  \tilde{\bfD}^{-1/2}$ which serves as a renormalisation trick to avoid numerical instability \parencite{kingma2014adam}. Moreover, the matrices $\bfW_0\in \mathcal{M}_{N \times d}(\mathbb{R})$ and $ \bfW_{\mu}, \bfW_{\sigma} \in \mathcal{M}_{d \times (Q-1)}(\mathbb{R}) $ correspond to the parameters to estimate, and $\bfH_1$ to the shared hidden layer. We denote the set of the VGAE parameters $\phi = \{\bfW_{0}, \bfW_{\mu}, \bfW_{\sigma}\}$ and the dimension of the latent space is set to $d=30$ in all our experiments. This encoder outputs the $(Q-1)$-dimensional mean vectors $(\mu_{\phi}(\bfA)_i)_i$ as well as the scalar log variances $(\ln \sigma_{\phi}(\bfA)_i)_i$. While other architectures might be of interest, in particular by using higher order neighbours, we choose to rely on the canonical GCN as it also recovers high order neighbour patterns \parencite{platonov2022critical} and limits the number of parameters. Eventually, the ELBO can be decomposed and computed such that:
\begin{equation}\label{eq:ELBO_detailed}
	\begin{split}
		\mathscr{L}(\bfPi; R_{\phi}) = & \sum_{j < i}  \mathbb{E}_{R_{\phi}(\bfZ)} \left[ \ln p(A_{ij} \mid \bm{\eta}_i, \bm{\eta}_j, \bfPi) \right]\\
		&  - \sum_{i=1}^N \KL\left(R(\bfZ_i) \mid p(\bfZ_i) \right).
	\end{split}
\end{equation}

\subsection{Identifiability}
In this section, we are interested in the identifiability of Deep LPBM. We consider a simplified version of the model by considering $\bfeta$ as a parameter, with $\bfeta$ the $N \times Q$ matrix with row $i$ corresponding to $\bfeta_i$ to get a better understanding of the issues that might be encountered. 

\begin{proposition}[\cite{daudin2010model}]\label{theorem:identifiability}
	Let $\bfeta$ be a $N \times Q$ matrix and $\bfPi$ a $Q \times Q$ such that: 
	\begin{itemize}
		\item[(H1)] $\bfeta = (\bfeta_1 \dots \bfeta_N)^{\top}$, where $\bfeta_i \in \Delta_{Q}$ for any $i \in \{1, \dots, N\}$.
		\item[(H2)] $\bfPi=(\bfPi_{qr})_{1 \leq q,r \leq Q}$ with $\bfPi_{qr} \in [0, 1]$.
	\end{itemize}
	Then, there exist $\tilde{\bfeta}$ and $\tilde{\bfPi}$ respecting (H1) and (H2) such that $ (\tilde{\bfeta}, \tilde{\bfPi}) \neq (\bfeta, \bfPi)$ and:
	\begin{equation}\label{eq:identifiability}
		\bfP = \tilde{\bfeta} \tilde{\bfPi} \tilde{\bfeta}^{\top} = \bfeta \bfPi \bfeta^{\top}.
	\end{equation}
	Therefore, the model is not identifiable. 
\end{proposition}

\begin{proof}
	Lemma 1 in the appendix gives sufficient conditions on a matrix $\bfH$ for $\tilde{\bfeta} = \bfeta \bfH$ and $\tilde{\bfPi} = \bfH^{-1} \bfPi (\bfH^{\top})^{-1}$ to respect \Cref{eq:identifiability}. Moreover, 
	the construction of such a matrix $\bfH$ is proposed in \cite{daudin2010model}. 
\end{proof}

In particular, the proposed construction of $\bfH$ modifies the variance of $\tilde{\bfPi}$ and the quantity $\Tr(\tilde{\bfeta}^{\top} \tilde{\bfeta})$. Another example from \cite{rubin2022statistical} shows that it is possible to construct a hyperbolic transformation of the partial memberships that may leave the matrix $\bfP$ unchanged. Concerning Deep LPBM, no estimation issue has been encountered in practice.

\subsection{Optimisation of the decoder}
The parameter $\bfPi$ and the variational parameters $\phi$ are updated using a stochastic gradient descent algorithm based on the reparametrisation trick. To use a gradient descent algorithm, the constrained $(\bfPi_{qr})_{q,r}$ in $]0,1[$ are mapped into the unconstrained set $\mathbb{R}$ using the following $f$ function $f(x) = 0.5 + \pi^{-1}\arctan(x)$ for any $x \in ]0,1[$. Hence, $f$ is a bijective mapping and denoting $\bfPi = (f(\tilde{\bfPi}_{qr}))_{1 \leq q,r \leq Q}$, we can optimise the parameters with respect to $\tilde{\bfPi}=(\tilde{\bfPi}_{qr})_{1\leq q,r \leq Q}$, a $Q \times Q$ unconstrained real matrix, using a gradient descent algorithm. For the sake of clarity, we denote the transformation $\bfPi = f(\tilde{\bfPi})$ the element-wise mapping of $\tilde{\bfPi}$ by $f$. In all our experiments, we used the Adam optimiser \parencite{kingma2014adam} with a learning rate set to $0.01$. Eventually, the model is estimated with $10$ different seeds and the one corresponding to the highest ELBO is kept as a result. The encoder is optimised using the reparametrisation \parencite{kingma2014autoencoding, rezende2014stochastic} and the entire procedure is shown in \Cref{algo:optim1}.

\newcommand\mycommfont[1]{\footnotesize\ttfamily\textcolor{blue}{#1}}
\SetCommentSty{mycommfont}
\begin{algorithm}
	\KwIn{$C^{\textbf{KMeans}}$ labels provided by a KMeans on $\bfA$\; $\bfZ^0 = f^{-1}(C^{\textbf{KMeans}})$\;}
	\tcc{Initialisation of the GNN}
	\For{epoch $\in \{1, \dots, \textrm{max iter$_{init}$}\}$}{
		$\vect{\mu_{\phi}}, \vect{\sigma_{\phi}}  \gets \enc(\bfA; \phi$)\;
		$\ell(\vect{\mu_{\phi}}, \vect{\sigma_{\phi}},  \bfZ^{0}) \gets \frac{1}{N} \sum_{i=1}^N \| \vect{\mu_{\phi, i}}- \vect{z}_i^0 \|_2^2  +
		\|\sigma_{\phi, i}^2 - 0.01 \|_2^2 $ \;
		Stochastic gradient descent on $\ell(\vect{\mu_{\phi}}, \vect{\sigma_{\phi}}, \bfZ^{0})$ with respect to $\phi$\;
	}
	\tcc{Estimation of Deep LPBM}
	\For{epoch $\in \{1, \dots, \text{max iter}\}$}{
		$\vect{\mu_{\phi}}, \vect{\sigma_{\phi}}  \gets \enc(\bfA; \phi$)\;
		$\bfZ \gets \vect{\mu_{\phi}} \oplus (\vect{\sigma_{\phi}} \odot \epsilon)$\;
		$\bfPi \gets f(\tilde{\bfPi})$\;
		$\hat{\vect{P}} \gets \dec(\bfZ; \bfPi)$\;	
		$\ell(\tilde{\bfPi}; \phi) \gets $ Using $\hat{\vect{P}}, \bfZ, \vect{\mu_{\phi}}$ and $ \vect{\sigma_{\phi}}$ in Equation \eqref{eq:ELBO}\;
		Stochastic gradient descent on $\ell(\tilde{\bfPi}; \phi)$ with respect to $\phi$ and $\tilde{\bfPi}$\;
	}
	\caption{Inference of the model parameters}
	\label{algo:optim1}
\end{algorithm}

\subsection{Model Selection}\label{sec:model_selection}
This section focuses on estimating the best number of clusters. To this aim, we compare the performance of Akaike's information criterion \parencite[AIC,][]{akaike1974new}, the Bayesian information criterion \parencite[BIC,][]{schwarz1978estimating} and the integrated classification likelihood \parencite[ICL,][]{biernacki1998assessing} to select the right number of clusters in Deep LPBM. As in \cite{daudin2010model}, $\bfeta$ is fixed and considered as a parameter. Hence, denoting $\zapfm$ the generative model, and $Q$ the fixed number of clusters considered, the three criteria can be computed as:
\begin{equation*}
	\begin{split}
		\operatorname{AIC}(Q, \zapfm) &= \ln p(\bfA \mid \widehat{\bfZ}, \widehat{\bfPi}) - \nu_{N,Q},\\
		\operatorname{BIC}(Q, \zapfm) &= \ln p(\bfA \mid \widehat{\bfZ}, \widehat{\bfPi}) - \frac{\nu_{N,Q}}{2} \ln \left( N_{obs} \right),\\
		\operatorname{ICL}(Q, \zapfm) &= \ln p(\bfA, \widehat{\bfZ} \mid \widehat{\bfZ}, \widehat{\bfPi}) - \frac{\nu_{N, Q, \bfPi}}{2} \ln \left( N_{obs} \right),\\
	\end{split}
\end{equation*}
where $\nu_{N,Q} = Q \left(Q+1\right) + N \left(Q-1\right)$ is the number of free parameters in the model, $\nu_{N, Q, \bfPi}~=~0.5 \ Q \left(Q+1\right)$ the number of free parameters in $\bfPi$,  and $N_{obs}~=~0.5 \ N \left(N-1\right)$ the number of observations. We point out that $\ln p(\bfZ)$, in the ICL, has an explicit form and does not depend on any parameters. The evaluation of these criterion is provided in \cref{sec:model_selection_evaluation}.

\section{Experiments on synthetic data}\label{sec:benchark}
This section aims at evaluating the efficiency of the proposed methodology. All the experiments were performed on a GPU NVIDIA P40 24Go. We are grateful to the Mésocentre Clermont-Auvergne of the Université Clermont Auvergne
for providing help, computing and storage resources.\footnote{Our code is available at: \url{https://anonymous.4open.science/r/deep_lpbm_package-1CCF/}} 

\subsection{Simulation settings}\label{sec:experiment_setting_sbm}
We start this section by presenting the settings and the underlying network structures responsible for our synthetic datasets. To appraise the efficiency of Deep LPBM, we sample undirected graphs made of $200$ and $500$ nodes and with $5$ clusters, described below and displayed in the appendix:
\begin{itemize}
	\item \textbf{Communities:} the probability of connection between nodes from the same cluster, denoted $\beta$, is higher than the probability of connection between nodes from different clusters, denoted $\varepsilon$
	\item \textbf{Disassortative:} the probability of connection between nodes from the same cluster, denoted $\varepsilon$, is lower than the probability of connection between nodes from different clusters, denoted $\beta$
	\item \textbf{Hub:} one of the clusters is highly connected to all the clusters, with a probability $\beta$, the other clusters are communities.
\end{itemize}

\paragraph{Sampling strategies to evaluate the node clustering efficiency of Deep LPBM}  On the one hand, we propose to assess the clustering efficiency of the proposed methodology, by sampling networks with three noise levels depending on $\beta$ equal to either $0.1$, $0.2$ or $0.3$. The higher $\beta$ is, the more structured the sampled network is and the easier it is to retrieve the true node partition. Note that the true $\bfeta_i$ are binary vectors here since each node belongs to a single cluster. As such, the neworks are not sampled from the model we consider. To evaluate the relevance of the estimated node partitions, the adjusted random index \parencite[ARI,][]{hubert1985comparing} is used to compute how close the estimated partition is to the true one. It is worth noticing that for methodologies rendering partial membership vectors, such as Deep LPBM, each node is assigned to its corresponding highest partial membership probability. The closer the ARI is to 1, the better the results are. A perfect retrieval of the cluster memberships gives an ARI of 1, while a random cluster assignment leads to an ARI of 0. We emphasise that computing an ARI on real data is not to be done in the context of \textbf{unsupervised learning} as illustrated in \Cref{section:real_data}.

\paragraph{Sampling strategies to evaluate the node partial memberships} On the other hand, the partial memberships estimation is evaluated by the following sampling setting. Let $\overline{\bfeta}_i$ be a one-hot encoded cluster membership and $\bfeta_{unif} = (1/Q \cdots 1/Q) \in \Delta_Q$ the vector corresponding to a uniform cluster membership, we introduce the variable $\zeta \in (0,1)$ to control the noise levels such that the true partial memberships are given by:
\[
\eta^{\star}_i = \zeta \overline{\bfeta}_i + (1 - \zeta) \bfeta_{unif}.
\]
The closer $\zeta$ is to $1$, the closer the sampling is to the stochastic block model generative assumptions. The results concerning the partial membership assignments are provided in \Cref{fig:partial_membership}. 
To evaluate the relevance of the estimated $\hat{\bfeta}$, we compare the amount of cluster membership shared between pairs of data points $\hat{\bfU} = \hat{\bfeta} \hat{\bfeta}^{\top}$ and the true ones $\bfU^{\star} = \bfeta^{\star} \bfeta^{\star\top}$. Inspired by \cite{heller2008statistical, latouche2014model}, we then compute the mean square-root error of the difference between the two matrices given by:
\begin{equation}\label{eq:metric_partial_membership}
	H = \sqrt{ \frac{2}{N (N-1)} \sum_{i \leq j} | \bfU^{\star}_{ij} - \hat{\bfU}_{ij} | }.
\end{equation}
Adjacency matrices sampled according to these two strategies, with different $\beta$ and $\zeta$ values, are represented in the appendix.

\subsection{Model selection}\label{sec:model_selection_evaluation}
In this section, we compare the performance of the three model selection criteria presented in \Cref{sec:model_selection}, namely the AIC, BIC and ICL. To this aim, we sampled $10$ networks with a true number of clusters $Q^{\star}$ equal to $5$, as detailed in Section 5. To evaluate these criteria, Deep LPBM is fitted with a number of clusters $Q$ equal to $2, 3, 4, 5, 6, 10,$ and $16$, such that for each number of clusters and each network, we run the methodology with $10$ different initialisations and keep the result corresponding to the highest ELBO. If a model is fitted with $Q$ clusters and one of them collapsed, meaning that no node belongs to it, the model is acknowledged as a $Q-1$ clusters model. \Cref{table:model_selection_criterion_full} presents the results of the three model selection criteria on the three network structures considered. In all cases, AIC outperforms BIC and ICL. AIC performs better than its two alternatives in all three network structures. In particular, Deep LPBM is able to recover the true number of clusters $100\%$ of the time, in the presence of communities, and within a disassortative structure, while BIC and ICL systematically collapse and select an under-parametrised model. For the hub structure, AIC shows a strong performance by selecting the right number of clusters $90\%$ of the times, while BIC and ICL again never select the right number of clusters. Consequently, we strongly advocate to use AIC as the Deep LPBM model selection criterion, as we shall do in the rest of this work.
\begin{table}
	\centering
	\caption{Comparison of AIC (\ref{table:AIC_comparison}), BIC (\ref{table:BIC_comparison}) and ICL (\ref{table:ICL_comparison}) to select the best number of clusters for Deep LPBM with $\beta$ equal to $0.3$. The line corresponding to the true number of clusters, equal to $5$, is highlighted and the most selected number of clusters is written in bold.}
	\label{table:model_selection_criterion_full}
	\begin{subtable}{0.3\linewidth}
		\centering
		\caption{AIC}
		\label{table:AIC_comparison}
		\begin{tabular}{LCCC}
			\toprule
			Q & \text{Com} & \text{Dis} & \text{Hub} \\
			\midrule
			1 & 0 & 0 & 0\\
			2 & 0 & 0 & 0 \\
			3 & 0 & 0 & 0 \\
			4 & 0 & 0 & 1 \\
			\rowcolor{shadowed}
			5^{\star} & \bm{10} & \bm{10} & \bm{9} \\
			6 & 0 & 0 & 0 \\
			10 & 0 & 0 & 0 \\
			16 & 0 & 0 & 0 \\
			\bottomrule
		\end{tabular}
	\end{subtable}
	\hfill
	\begin{subtable}{0.3\linewidth}
		\centering
		\caption{BIC}
		\label{table:BIC_comparison}
		\begin{tabular}{LCCC}
			\toprule
			Q & \text{Com} & \text{Dis} & \text{Hub} \\
			\midrule
			1 & 0 & 2 & 0\\
			2 & \bm{10} & \bm{8} & \bm{10}\\
			3 & 0 & 0 & 0 \\
			4 & 0 & 0 & 0 \\
			\rowcolor{shadowed}
			5^{\star} &	0 & 0 & 0 \\
			6 & 0 & 0 & 0 \\
			10 & 0 & 0 & 0 \\
			16 & 0 & 0 & 0 \\
			\bottomrule
		\end{tabular}
	\end{subtable}
	\hfill
	\centering
	\begin{subtable}{0.3\linewidth}
		\centering
		\caption{ICL}
		\label{table:ICL_comparison}
		\begin{tabular}{LCCC}
			\toprule
			Q & \text{Com} & \text{Dis} & \text{Hub} \\
			\midrule
			1 & 0 & 2 & 0 \\
			2 & \bm{10} & \bm{8} & \bm{9}\\
			3 & 0 & 0 & 1 \\
			4 & 0 & 0 & 0 \\
			\rowcolor{shadowed}
			5^{\star} &	0 & 0 & 0 \\
			6 & 0 & 0 & 0 \\
			10 & 0 & 0 & 0 \\
			16 & 0 & 0 & 0 \\
			\bottomrule
		\end{tabular}
	\end{subtable}
\end{table}

\subsection{Benchmark evaluating the clustering performances}

This section aims to evaluate Deep LPBM as a node clustering methodology. Let us recall that Deep LPBM generative model is designed to estimate node partial memberships and not ``hard'' cluster assignments, contrary to SBM. Therefore, we associate each node to the cluster corresponding to its highest partial membership. We use the settings described in \Cref{sec:experiment_setting_sbm}. We stress that this sampling scheme corresponds to SBM generative assumptions, and, as such, Deep LPBM is not favoured by the sampling scheme compared to its competitors. We evaluate the model against the adversarially regularised variational graph autoencoder \parencite[ARVGA,][]{pan2018adversarially}, against the variational graph autoencoder \parencite[VGAE,][]{kipf2016variational}. We also compare Deep LPBM against the deep latent position cluster model \parencite[DLPM,][]{liang2022deep}, the stochastic block model \parencite[SBM,][]{holland1983stochastic, daudin2008mixture}, with  random initialisation (SBM random) and K-Means initialisation (SBM kmeans) as well as the variational Bayes latent position cluster model \parencite[VBLPCM,][]{salter2013variational}. For the methodologies that do not perform node clustering, namely the VGAE and the ARGVA, a K-Means algorithm is fitted on the estimated posterior node embeddings with the true number of clusters. More details concerning those competitors are provided in the appendix.
All methodologies are estimated using the true number of clusters, and the results are reported in \Cref{tab:benchmark_ari}. The best results is coloured in \red{red}, the second best in \blue{blue} and the third one in \green{green}. When two results are equal up to the standard deviation, they are identically coloured and if no signal is recovered, i.e with an ARI too low, no colouration is used.

In the absence of noise, corresponding to $\beta$ equal to $0.3$, Deep LPBM is the only positional methodology able to perfectly recover the true partition of the nodes in all three network structures. SBM with a K-Means initialisation also recovers the true partitions. However, with random initialisation, SBM obtains ARI significantly lower than Deep LPBM in all structures. Without noise, Deep LPBM outperforms all the positional methodologies and is as good as the SBM with a K-Means initialisation, which is specifically designed for the task of node clustering.

For $\beta$ equal to $0.2$, both in the community case and the hub case, Deep LPBM continues to efficiently cluster the nodes. It recovers the node partitions almost perfectly in the community structure. It obtains the second best results behind SBM with a K-Means initialisation. In the hub structure, it reaches an ARI of $0.89$, the second-best ARI, performing as well as DLPM, behind SBM with a K-Means initialisation. The disassortative case makes it difficult to simultaneously obtain the node positions and estimate block connectivity, leading to an ARI of $0.39$. While this is far behind the results of SBM with a K-Means initialisation and closely behind to SBM with a random initialisation, it is still an improvement upon all positional methods that cannot retrieve any signal.

\begin{table}[H]
	\centering
	\small
	\caption{Benchmark to compare Deep LPBM with competitors on three different graph structures, namely Communities (Com), Disassortative (Dis) and Hub structures, with three different noise levels.}
	\label{tab:benchmark_ari}
	\begin{tabular}{llCCC}
		\toprule
		&  & \text{Com} & \text{Dis} & \text{Hub} \\
		\midrule
		\multirow[c]{5}{*}{\rotatebox{90}{$\beta = 0.2$}}\hspace*{-0.5cm} & VBLPCM & \green{0.98 \pm 0.02} & 0.01 \pm 0.00 & 0.72 \pm 0.15 \\
		& DLPM & \blue{0.99 \pm 0.01} & 0.00 \pm 0.00 & \blue{0.89 \pm 0.10} \\
		& ARVGA & 0.85 \pm 0.03 & 0.01 \pm 0.01 & 0.28 \pm 0.06 \\
		& VGAE & 0.97 \pm 0.02 & 0.00 \pm 0.01 & 0.64 \pm 0.23 \\
		& SBM kmeans & \red{1.00 \pm 0.01} & \red{1.00 \pm 0.01} & \red{0.95 \pm 0.10} \\
		& SBM random & 0.70 \pm 0.03 & \blue{0.45 \pm 0.19} & 0.82 \pm 0.16 \\
		& Deep LPBM & \blue{0.99 \pm 0.01} & \green{0.39 \pm 0.13} & \blue{0.89 \pm 0.09} \\
		\midrule
		\multirow[c]{5}{*}{\rotatebox{90}{$\beta = 0.3$}}\hspace*{-0.5cm} & VBLPCM & \red{1.00 \pm 0.00} & 0.01 \pm 0.01 & 0.79 \pm 0.13 \\
		& DLPM & \red{1.00 \pm 0.00} & 0.00 \pm 0.00 & \blue{0.98 \pm 0.01} \\
		& ARVGA & 0.88 \pm 0.03 & 0.06 \pm 0.04 & 0.56 \pm 0.22 \\
		& VGAE & \red{1.00 \pm 0.00} & 0.00 \pm 0.01 & 0.72 \pm 0.16 \\
		& SBM K init & \red{1.00 \pm 0.00} & \red{1.00 \pm 0.00} & \red{1.00 \pm 0.00} \\
		& SBM R init & 0.68 \pm 0.15 & \blue{0.79 \pm 0.17} & \green{0.94 \pm 0.13} \\
		& Deep LPBM & \red{1.00 \pm 0.00} & \red{1.00 \pm 0.00} & \red{1.00 \pm 0.01} \\
		\bottomrule
	\end{tabular}
\end{table}

\subsection{Evaluation as a partial memberships model}\label{seq:partial_membership_evaluation}

The goal of this section is to assess Deep LPBM efficiency to estimate node partial memberships. Since VGAE and ARVGA do not provide such a quantity, they cannot be used for comparison. The $H$-quantity described in \Cref{eq:metric_partial_membership} is computed and averaged over the $10$ sampled networks. The results are provided in \Cref{fig:partial_membership}, where the shadowed area corresponds to the standard deviation computed from the estimations over the $10$ networks. Note that the sampling scheme used for each network is different from Deep LPBM generative assumptions, and as such does not favour the proposed methodology.

We start by remarking that $\zeta$ inferior or equal to $0.5$ induces partial memberships $\bfeta$ closer to a uniform membership distribution than to a one-hot encoded vector. In all three structures, SBM is not able to translate this contrary to Deep LPBM, which is indicated by the gap between the Deep LPBM and SBM results, for small $\zeta$ values. Interestingly, the only time SBM becomes slightly better than Deep LPBM, in the disassortative structure, happens for $\zeta$ getting closer to $1$, meaning when sampling assumptions are getting closer to SBM generative model. In all other cases, Deep LPBM renders more accurate partial memberships than SBM. In addition, Deep LPBM outperforms all positional methodologies. Even though DLPM can translate the uniform distribution of $\bfeta$ for a disassortative structure, it is clear that it is due to the absence of signal to detect, as can be seen on Figure 2 of the appendix. Since DLPM fails to detect any signal for a network with a disassortative structure, as stated in the previous section, the results for low $\zeta$ values only indicate the absence of detected signal. In particular, as soon as $\zeta$ increases, the metric worsens and approaches $0.6$, as for VBLPCM. This indicates that both methods fail to estimate $\bfeta$. Overall, since Deep LPBM provides better estimates regarding the node partial membership than all tested methodologies. 

\begin{figure}
	\centering
	\includegraphics[width=1\linewidth]{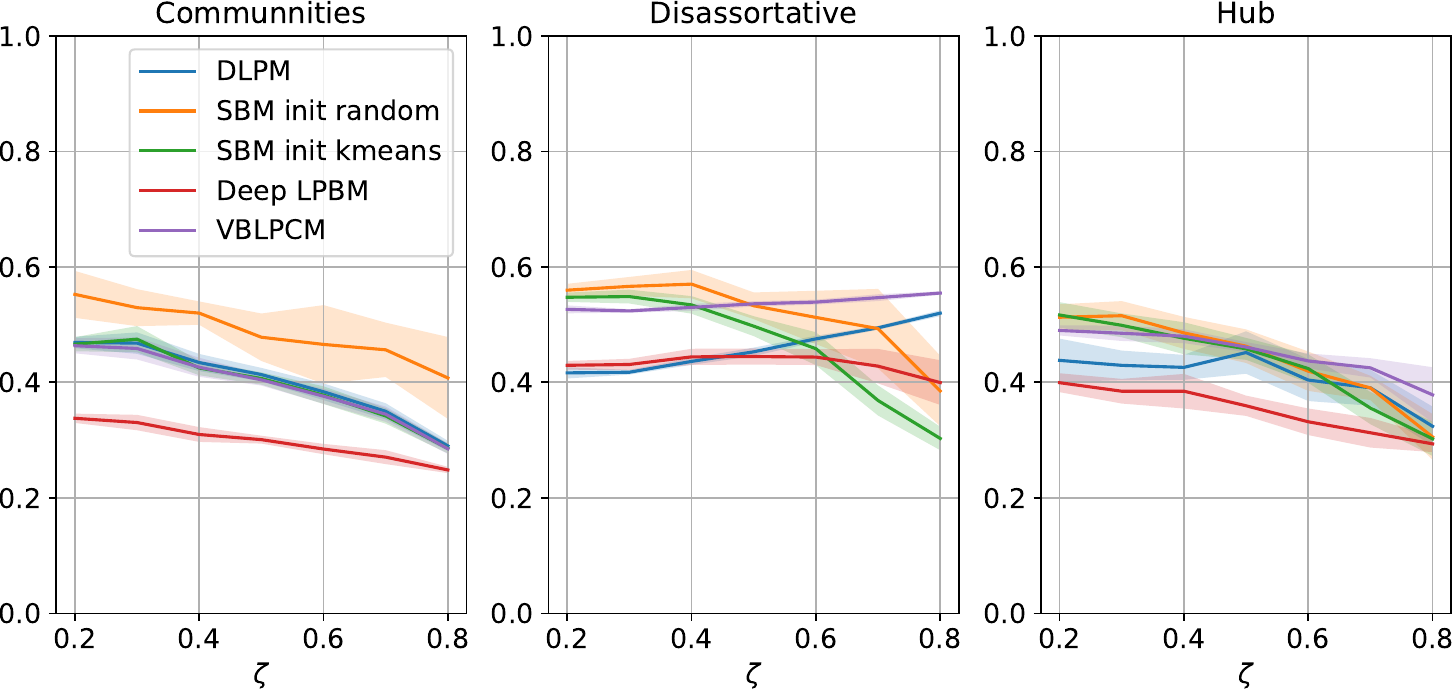}
	\caption{The partial memberships evaluation is obtained by plotting the results of \cref{eq:metric_partial_membership} for different values of $\zeta$. The lower the quantity is, the better the estimation of $\bfeta$ is. A larger version is available in the appendix.}
	\label{fig:partial_membership}
\end{figure}

\section{Analysis of the French political blogosphere network}\label{section:real_data}

In this last section, we propose to apply Deep LPBM on the  ``French Political Blogosphere'' dataset \parencite{zanghi2008fast}. This network was collected in 2006 in order to analyse the French presidential campaign on the web. Each node corresponds to a blog, and each edge to a hyperlink between two blogs. We assume that the edges are not oriented to simplify the analysis. In addition, when several hyperlinks are present between two nodes, they are gathered into a single edge. The number of clusters selected by AIC is equal to $8$, see the appendix for the evolution of AIC in function of $Q$. \Cref{fig:blog_network_visualisation} is obtained by projecting $\bfeta$, estimated by Deep LPBM with the t-sne algorithm \parencite{van2008visualizing}. The node colours in the left-hand side figure denote the political party associated to the corresponding blogs, while the node pie-charts, in the right-hand side figure, represent the estimated partial cluster memberships $\bfeta$.  In addition, $\hat{\bfPi}$ is displayed in \Cref{fig:estimated_pi_matrix_blog}.

We start by noting that in \Cref{fig:blog_network_visualisation_dlpbm}, Cluster $5$ corresponds to poorly connected nodes as shown by $\hat{\bfPi}$ in \Cref{fig:estimated_pi_matrix_blog}. Thanks to the partial memberships, we can refine these results by noting that among the poorly connected nodes, some partial memberships hold several colours, indicating that their connections share the connectivity patterns of several clusters. To give an example, among the UDF blogs in \Cref{fig:blog_network_visualisation_political_parties}, some are poorly connected to other nodes, as the nodes in Cluster $5$, but also share connectivity patterns with Cluster $6$. Thanks to the partial memberships, they are placed in between Clusters $5$ and $6$ in \Cref{fig:blog_network_visualisation_dlpbm}, with pie charts indicating a significant level for the two partial memberships, namely the red and the violet one. This observation can be extended to the rest of the network.

Cluster 2 presents interesting properties displayed in \Cref{fig:U_hat_blog}. Indeed, its first characteristic is to be highly internally connected as indicated by the $(2,2)$ coordinate of the $\bfPi$ matrix displayed in \Cref{fig:estimated_pi_matrix_blog}. In addition, nodes behaving like Cluster $2$ tend to also share similar connectivity patterns with other clusters, as shown by the correlations displayed in \Cref{fig:U_hat_blog}. Indeed, many nodes with a high Cluster $2$ partial membership seem to behave similarly to Clusters $4$ and $5$ and, to a lesser extent, to Cluster $1$, $6$ and $8$. In other words, these nodes behave almost like a hub. This flexibility can only be obtained at the node level and cannot be rendered at the block level. Therefore, the visualisation of the entire network is enriched compared to the visualisation based only on the clusters, as in block modelling. Indeed, in \Cref{fig:blog_network_visualisation_dlpbm}, many nodes holding a significant partial membership to Cluster 2 are in between clusters.

To conclude, let us emphasise that comparing ARI in an unsupervised setting on real data is not helpful to evaluate the effectiveness of the methodology. To illustrate this, let us consider the political parties as the  ``true node labels''. It would imply that Cluster $1$ and $7$ should be merged into a single group. However, \Cref{fig:Pi_and_eta_matrices} clearly shows that two connectivity patterns compose UMP blogs, represented by Cluster $1$ and $7$, poorly connected one to another but highly connected internally. Indeed, unsupervised learning aims at finding patterns within the data and not to predict labels. A comparison with SBM results is provided in the appendix.
\begin{figure}[t]
	\centering
	\subcaptionbox{Estimated $\hat{\bfPi}$.	\label{fig:estimated_pi_matrix_blog}}[0.42\linewidth]{	\includegraphics[width=\linewidth]{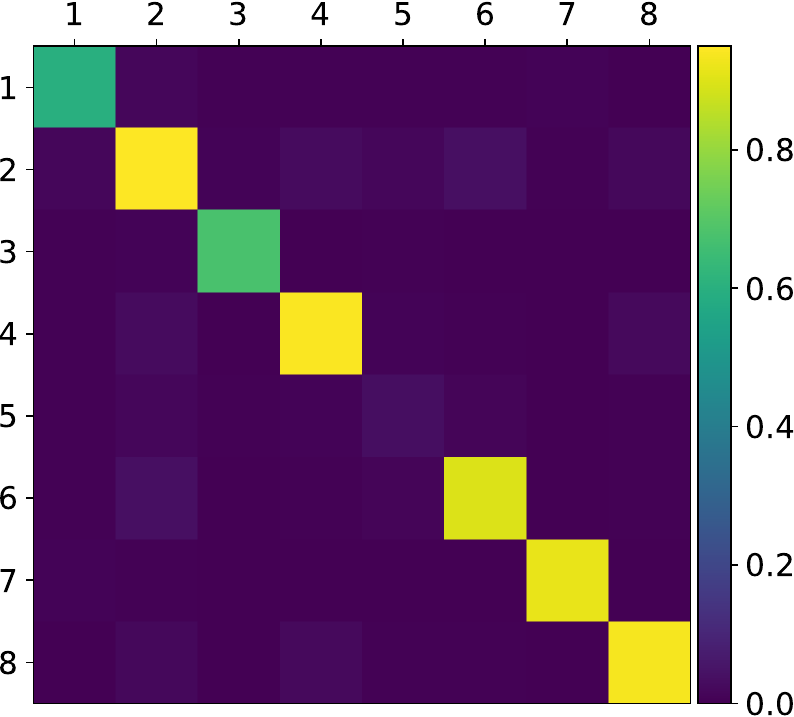}}
	\subcaptionbox{Estimated $\hat{\bfU}$ defined in \Cref{sec:experiment_setting_sbm}.		\label{fig:U_hat_blog}}[0.42\linewidth]{\includegraphics[width=\linewidth]{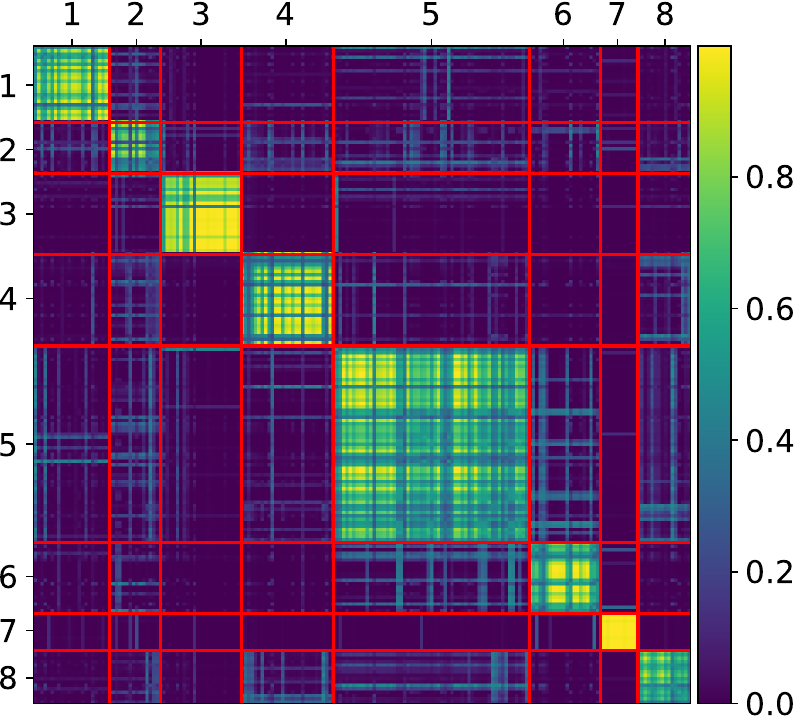}}
	\caption{Visualisation of $\hat{\bfPi}$ and $\hat{\bfU}$ matrices. On the right-hand side, $\hat{\bfU}$ is a $N \times N$ matrix but is ordered by block which are delimited by the red lines.}
	\label{fig:Pi_and_eta_matrices}
\end{figure}
\vspace*{-0.1cm}
\begin{figure}[t]
	\centering
	\subcaptionbox{Political parties.\label{fig:blog_network_visualisation_political_parties}}[0.49\linewidth]{	\includegraphics[width=\linewidth]{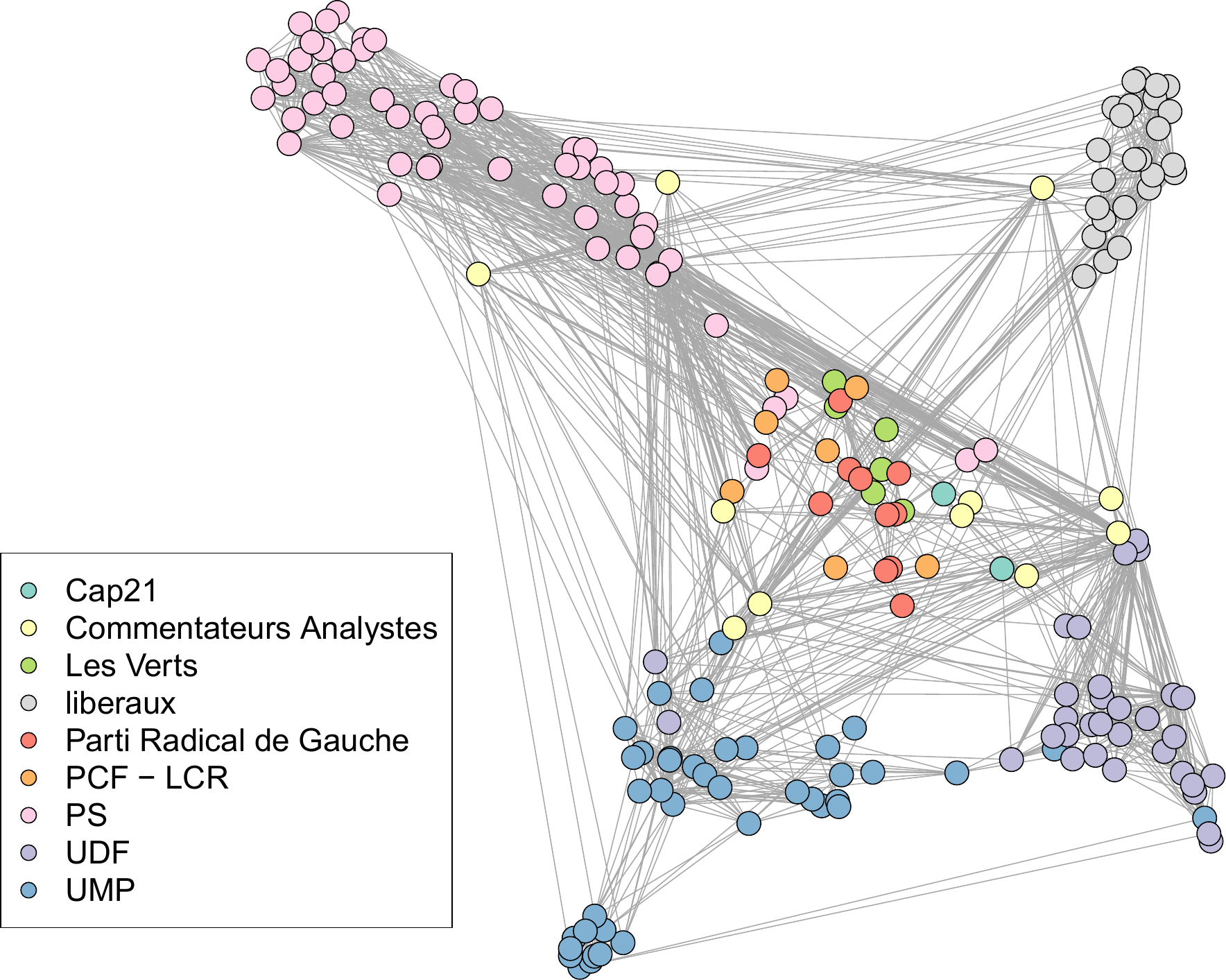}}
	\subcaptionbox{Deep LPBM partial memberships.		\label{fig:blog_network_visualisation_dlpbm}}[0.49\linewidth]{\includegraphics[width=\linewidth]{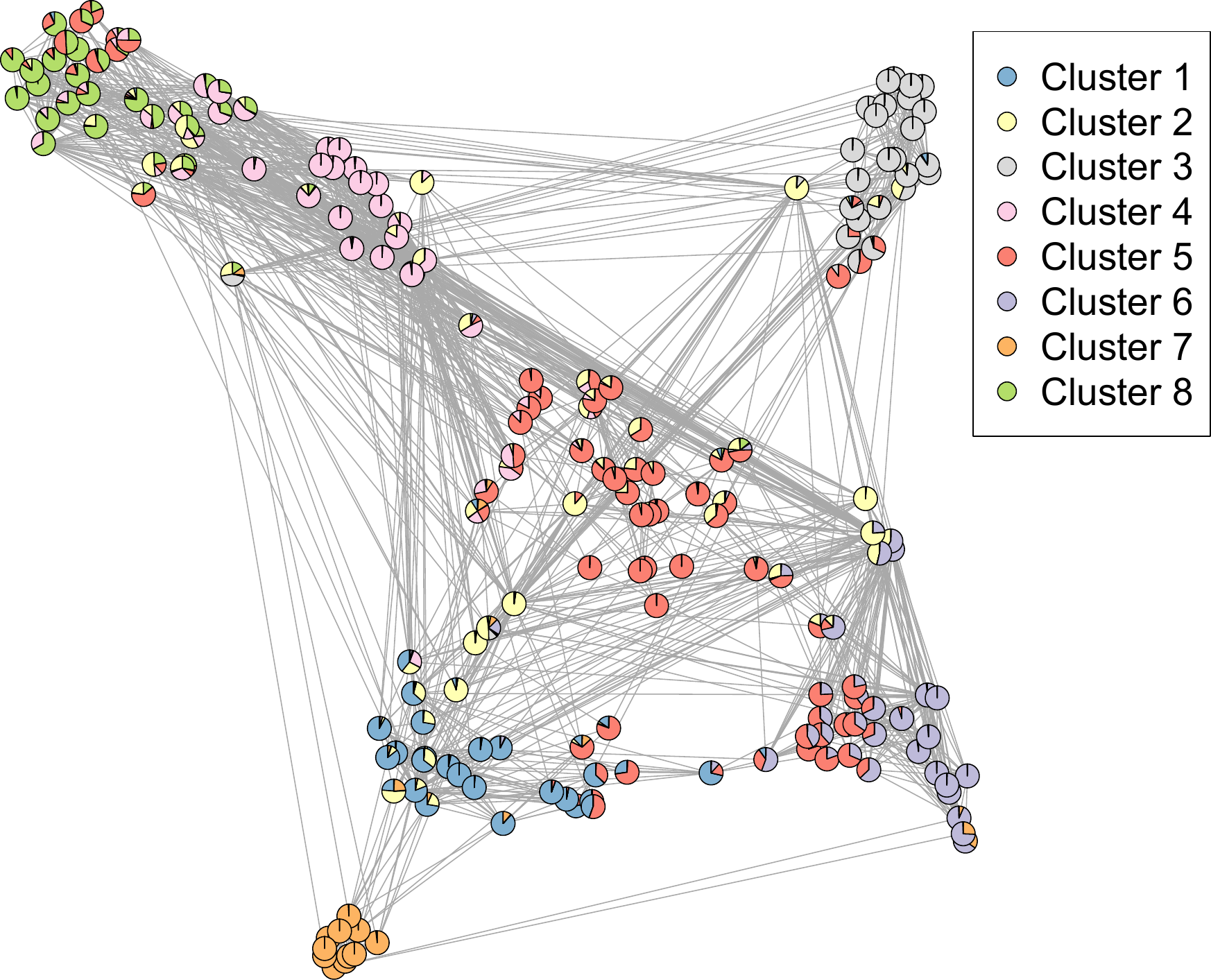}}
	\caption{All node positions were estimated by Deep LPBM. On the right-hand side, the node colours indicate the political party associated to the blog. On the left-hand side, each node is depicted by a pie chart representing the partial memberships estimated by Deep LPBM. A large version is available in the appendix.}
	\label{fig:blog_network_visualisation}
\end{figure}

\section{Conclusion and discussion}
\label{sec:conclusion}

We introduced a novel end-to-end methodology that simultaneously estimates node partial memberships, cluster connectivity patterns, as well as a cluster-based network visualisation. This framework, entitled the deep latent position block model (Deep LPBM), provides refined results compared to block model approaches. Contrary to those methods, it is also able to provide a visualisation of the entire network based on the estimation of the partial memberships. Deep LPBM extends the current position-based methodologies by using a variational graph autoencoder (VGAE) with a specifically designed new block decoder, allowing to analyse new connectivity patterns, such as disassortative networks. In addition, an extensive benchmark of Deep LPBM against state-of-the art methods is provided to first assess the quality of the estimated partial memberships and second to evaluate the clustering performances of the algorithm. Deep LPBM outperforms all competitors on the first task and shows very good results in low and moderate noise regimes on the second task. Eventually, a quality evaluation of the results is provided through a comparison of Deep LPBM and SBM results on the French political blogosphere.

\printbibliography

\newpage
\appendix

\section{Inference}

\subsection{Identifiability}
The following lemma permits to construct two different sets of parameters resulting in the same probabilities of connection.
\begin{lemma}[\cite{daudin2010model}]\label{lemma:H_matrix_identifiability}
	Let $\bfH \in \mathcal{M}_{Q \times Q}(\mathbb{R})$ be a matrix such that:
	\begin{enumerate}[label=(\arabic*)]
		\item[(A1)] $\bfH^{-1}$ exists,
		\item[(A2)] $\bfH \bm{1_Q} = \bm{1}_Q$, where $\bm{1}_Q = (1, \dots, 1)^{\top}$ be the $Q$-dimensional vector made of $1$, 
		\item[(A3)] $\tilde{\bfeta} = \bfeta \bfH \geq 0$,
		\item[(A4)] $\tilde{\bfPi} = \bfH^{-1} \bfPi (\bfH^{\top})^{-1} \in \mathcal{M}_{Q \times Q}([0, 1])$.
	\end{enumerate}
	Then, the following holds true:
	\begin{itemize}
		\item For any node $i$, $\tilde{\bfeta}_i^{\top} \bf{1}_Q = \bfeta_i^{\top} \bfH \bf{1}_Q =  \bfeta_i \bf{1}_Q = 1$, i.e $\bfeta_i \in \Delta_Q$,
		\item $\tilde{\bfPi} \in \mathcal{M}_{Q \times Q}([0, 1])$,
		\item $\tilde{\bfeta} \tilde{\bfPi} \tilde{\bfeta}^{\top} = \bfeta \bfH \bfH^{-1} \bfPi (\bfH^{\top})^{-1} \bfH^{\top} \bfeta^{\top} =  \bfeta  \bfPi \bfeta^{\top}$.
	\end{itemize}
\end{lemma}
\begin{proof}
	The proof is a direct application of the assumptions (A1) to (A4).
\end{proof}

\section{Experiments on synthetic data}\label{sec:synthetic_data_appendix}

\subsection{Simulation settings}\label{sec:simu_settings_appendix}
In this section, we provide additional information concerning the simulation setting proposed to establish the benchmark. \Cref{tab:PI_matrices} presents the connectivity matrices corresponding to the network structure considered in the simulations.

\begin{table}[H]
	\centering
	\begin{tabular}{c}
		Communities \\
		$
		\begin{pmatrix}
			\beta & \varepsilon & \varepsilon & \varepsilon & \varepsilon \\
			\varepsilon & \beta & \varepsilon & \varepsilon & \varepsilon \\
			\varepsilon & \varepsilon & \beta & \varepsilon & \varepsilon \\
			\varepsilon & \varepsilon & \varepsilon & \beta & \varepsilon \\
			\varepsilon & \varepsilon & \varepsilon & \varepsilon & \beta
		\end{pmatrix}
		$
	\end{tabular}
	\begin{tabular}{c}
		Disassortative \\
		$
		\begin{pmatrix}
			\varepsilon & \beta & \beta & \beta & \beta \\
			\beta & \varepsilon & \beta & \beta & \beta \\
			\beta & \beta & \varepsilon & \beta & \beta \\
			\beta & \beta & \beta & \varepsilon & \beta \\
			\beta & \beta & \beta & \beta & \varepsilon
		\end{pmatrix}
		$ 
	\end{tabular}
	\begin{tabular}{c}
		Hub \\
		$
		\begin{pmatrix}
			\beta & \beta & \beta & \beta & \beta \\
			\beta & \beta & \varepsilon & \varepsilon & \varepsilon \\
			\beta & \varepsilon & \beta & \varepsilon & \varepsilon \\
			\beta & \varepsilon & \varepsilon & \beta & \varepsilon \\
			\beta & \varepsilon & \varepsilon & \varepsilon & \beta
		\end{pmatrix}
		$
	\end{tabular}
	
	\label{tab:PI_matrices}
	\caption{Connectivity matrix $\bfPi$ corresponding to a community, a disassortative, and a hub network structure.}
\end{table}

\begin{table}
	\centering
	\newcommand{\spacing}{0.1}
	\newcommand{\FigWidth}{0.22}
	\begin{tabular}{lccc}
		& Communities & Disassortative & Hub\\
		\rotatebox{90}{\hspace*{\spacing cm}$\beta = 0.2$} & 
		\includegraphics[width=\FigWidth \linewidth]{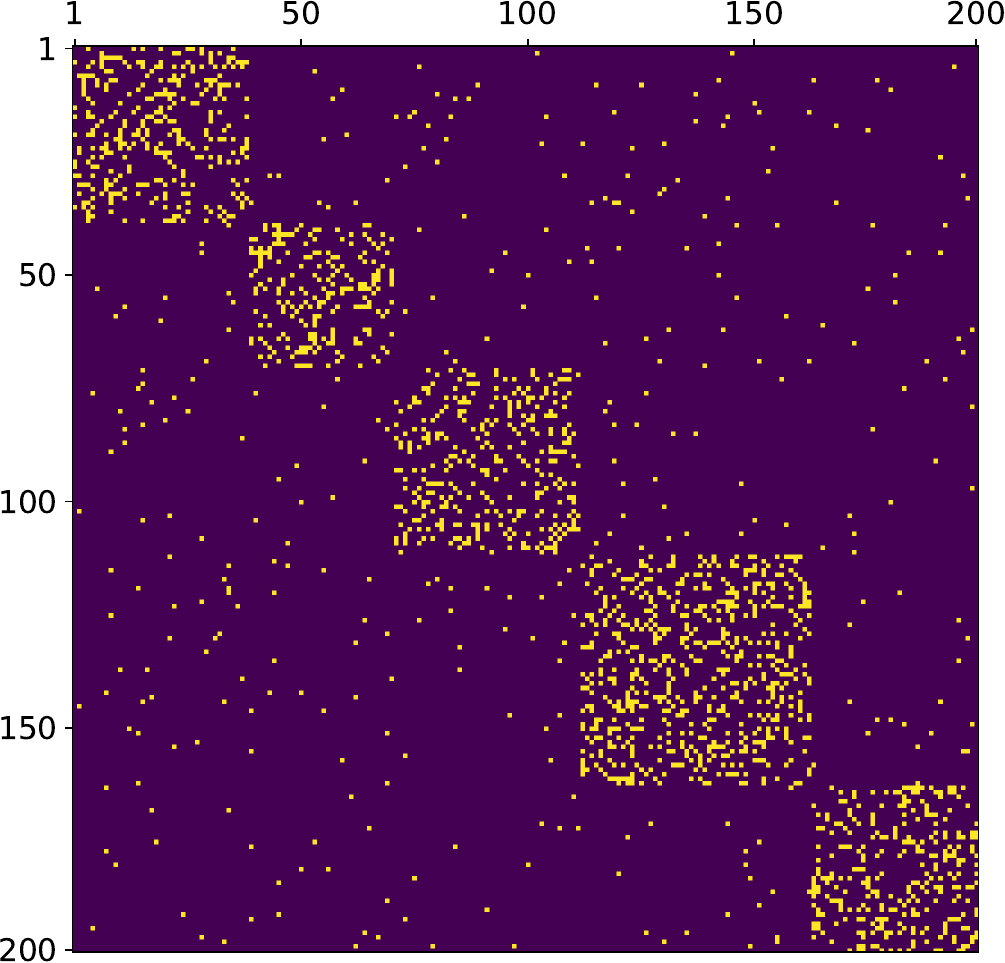}&
		\includegraphics[width=\FigWidth \linewidth]{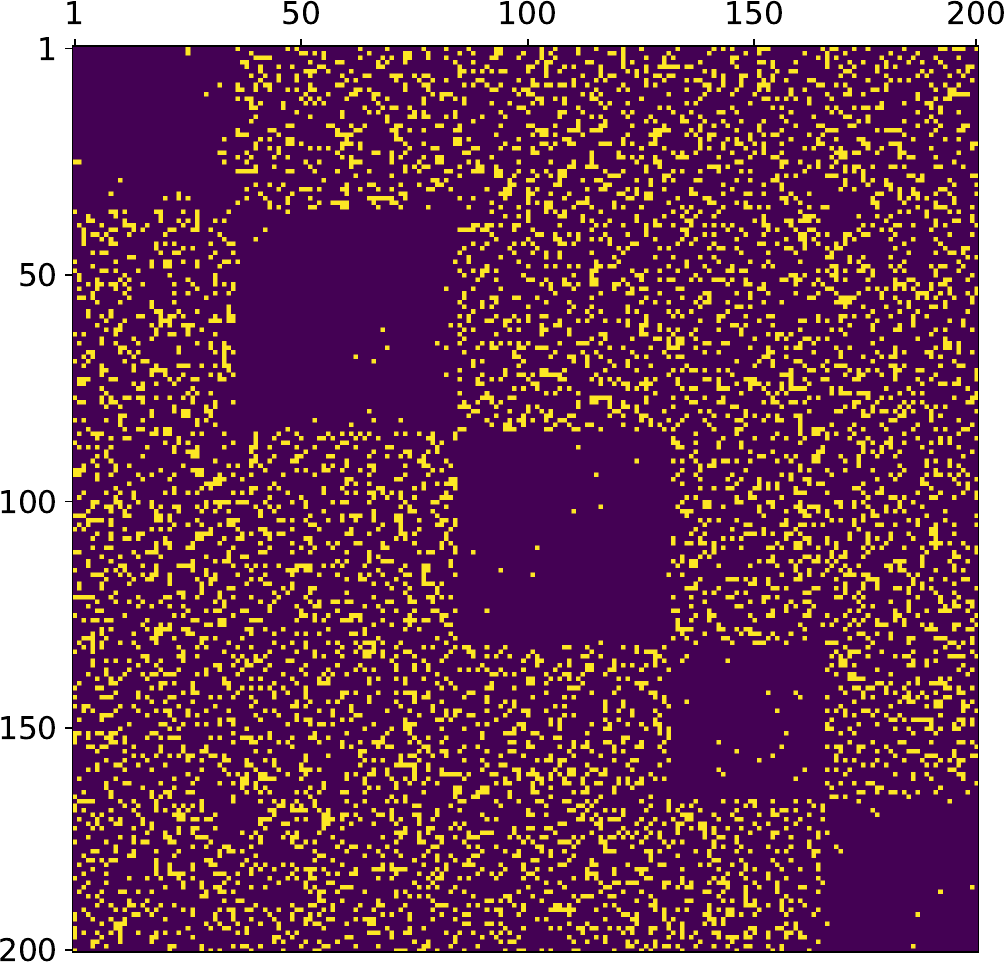}&
		\includegraphics[width=\FigWidth \linewidth]{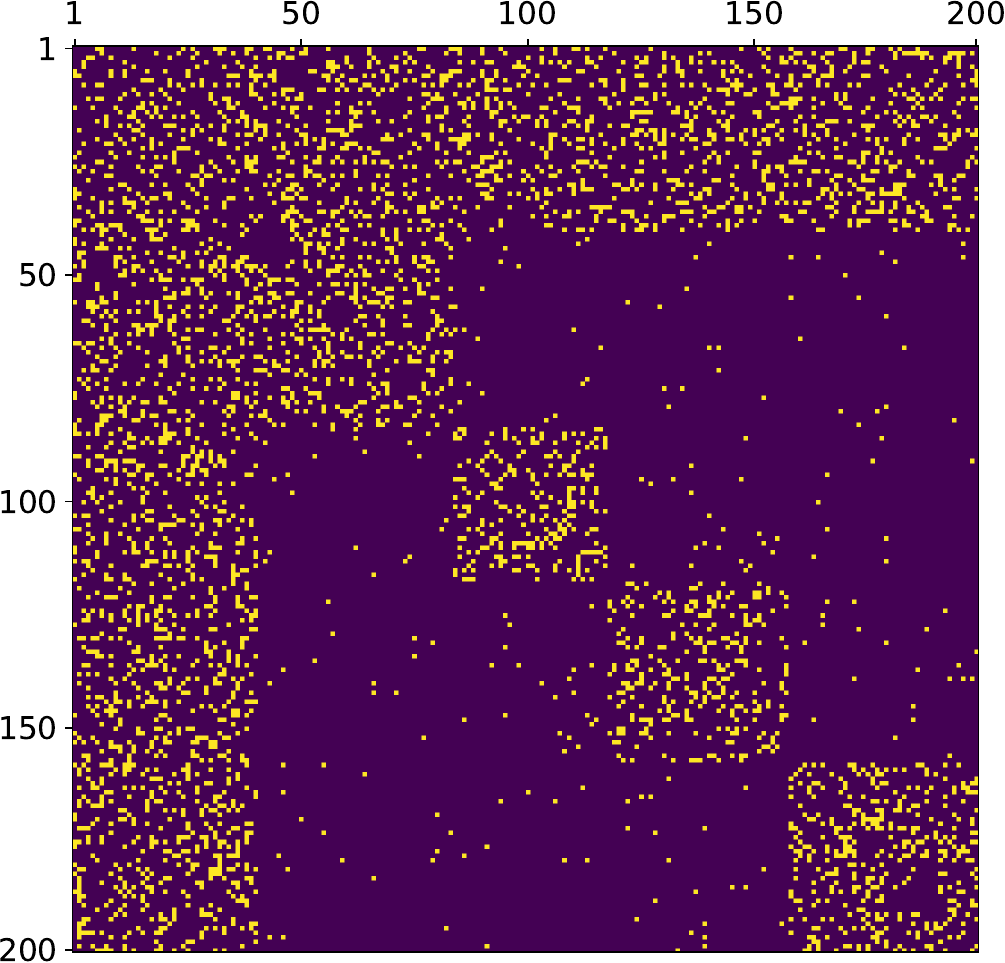}\\
		\rotatebox{90}{\hspace*{\spacing cm}$\beta = 0.3$} & 
		\includegraphics[width=\FigWidth \linewidth]{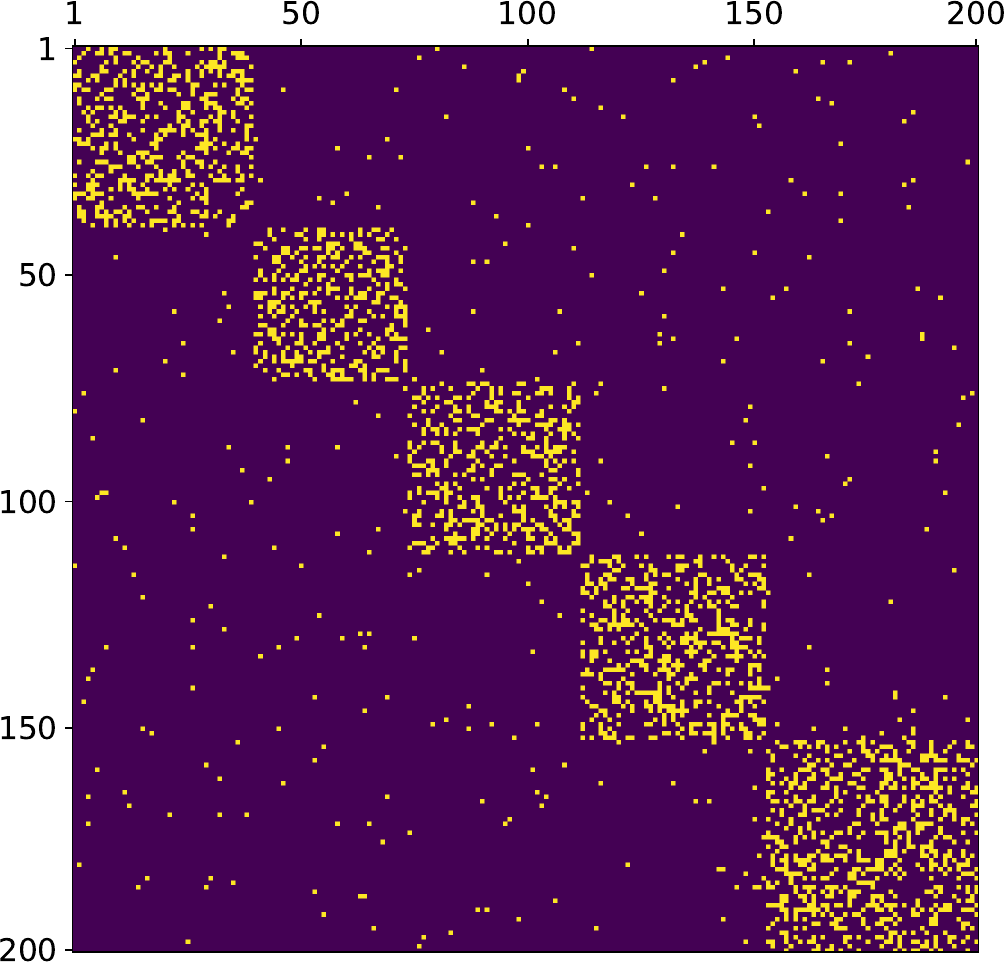}&
		\includegraphics[width=\FigWidth \linewidth]{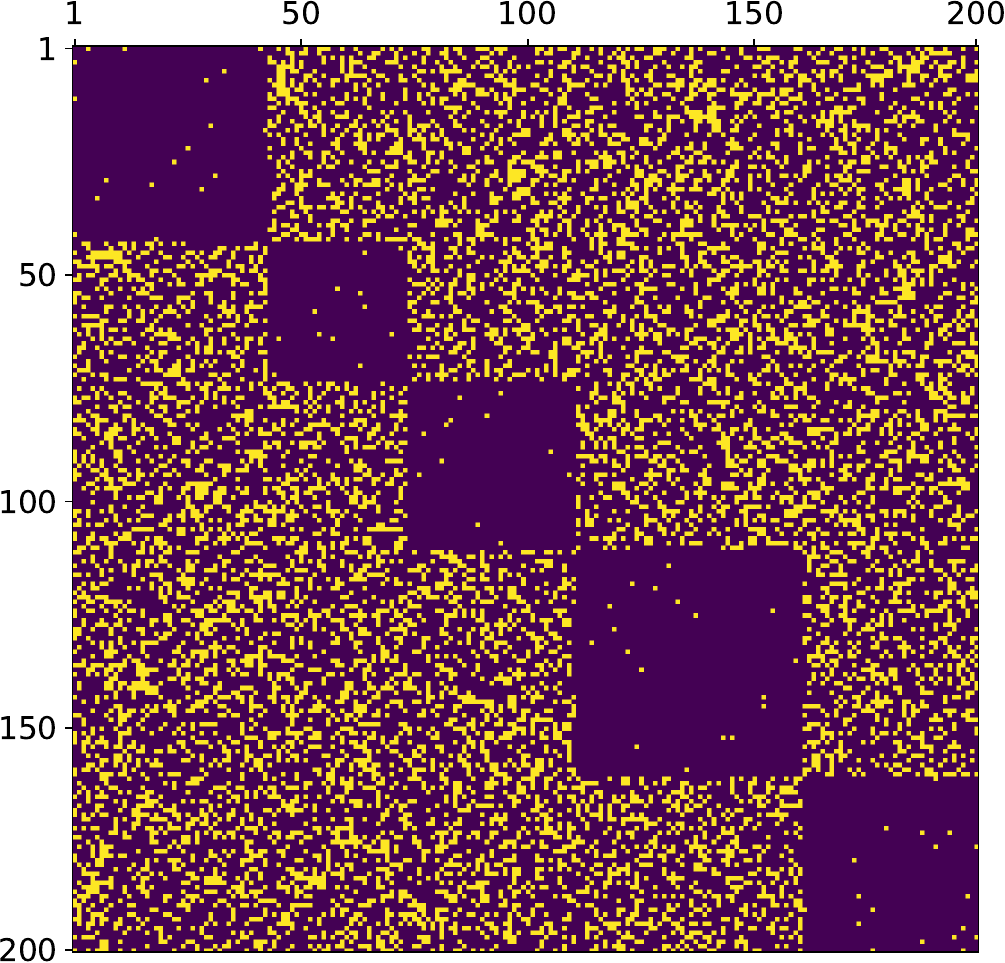}&
		\includegraphics[width=\FigWidth \linewidth]{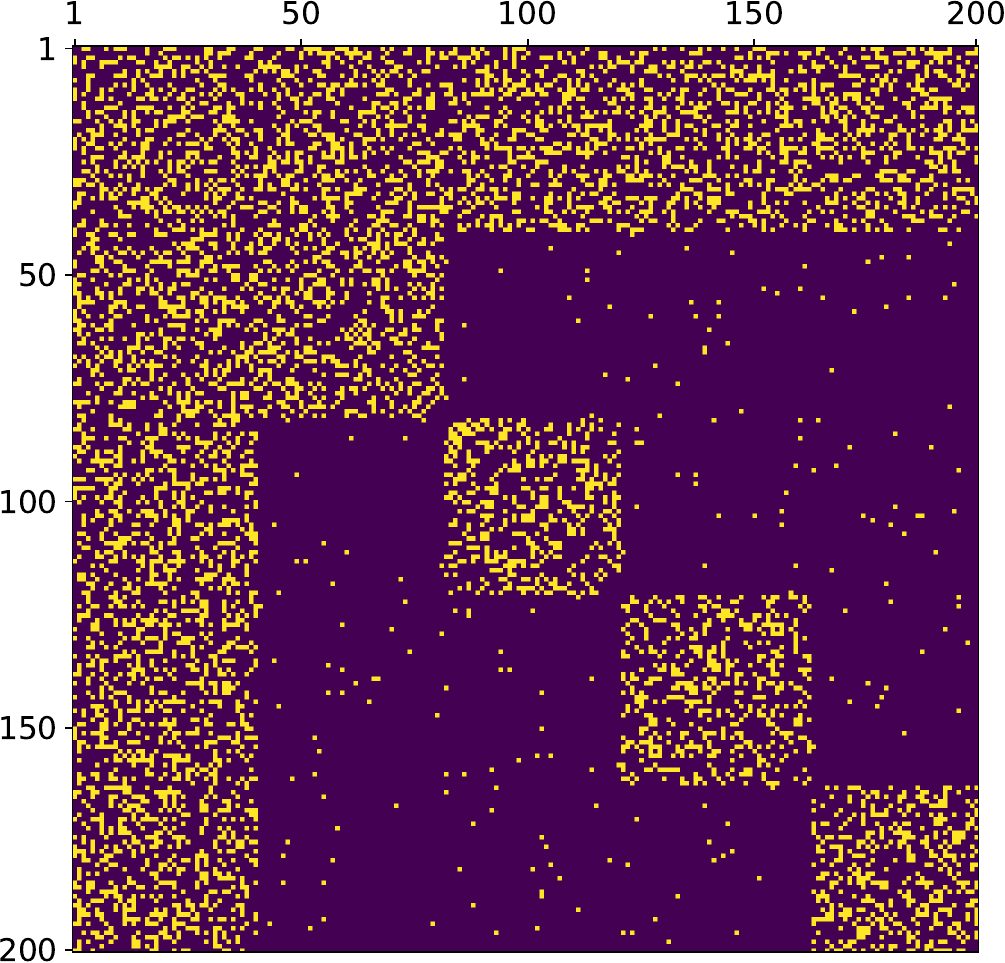}\\	
	\end{tabular}
	\captionof{figure}{Adjacency matrix sampled according to the clustering setting, as described in Section 5 for different values of $\beta$. A yellow (purple respectively) pixel at coordinate $(i,j)$ indicates the existence (the absence) of an edge between nodes $i$ and $j$.}
	\label{fig:appendix_adjacency_clustering_details}
\end{table}

\begin{table}
	\centering
	\newcommand{\spacing}{0.23}
	\newcommand{\FigWidth}{0.22}
	\begin{tabular}{lccc}
		& Communities & Disassortative & Hub\\
		\rotatebox{90}{\hspace*{\spacing cm}$\zeta = 0.2$} & 
		\includegraphics[width=\FigWidth \linewidth]{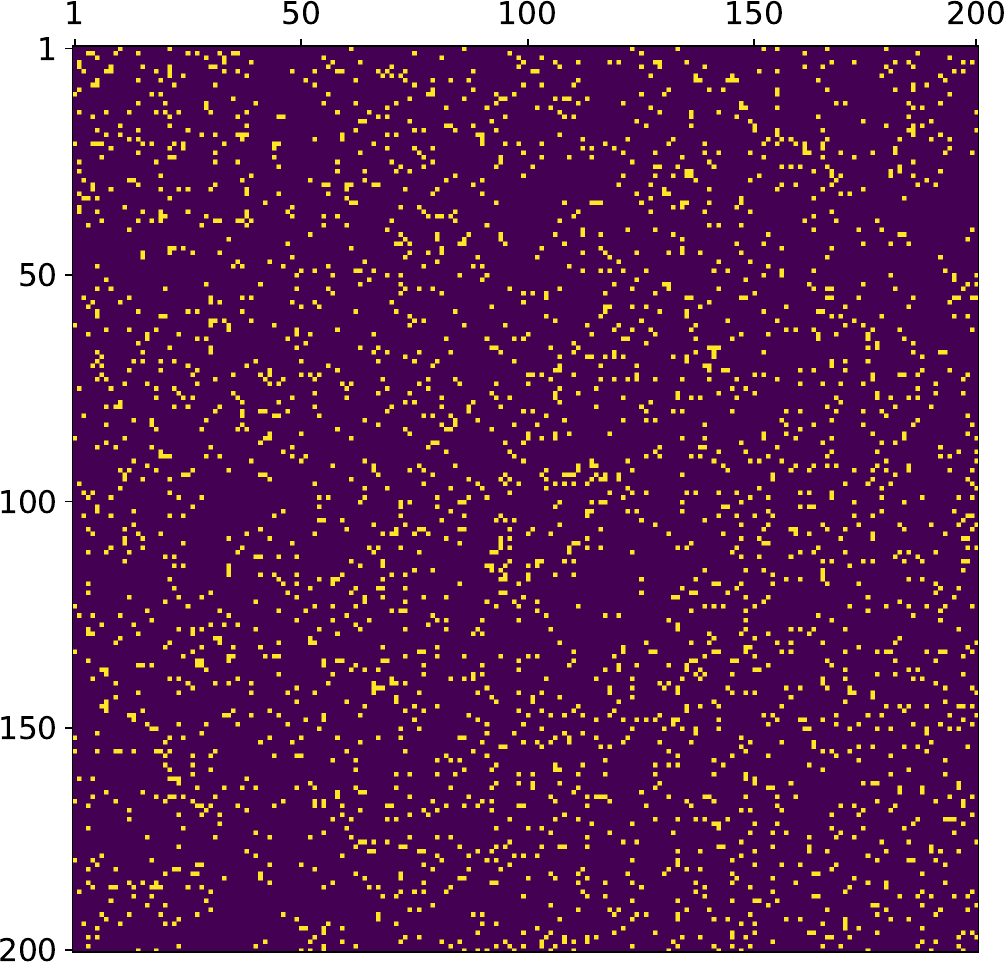}&
		\includegraphics[width=\FigWidth\linewidth]{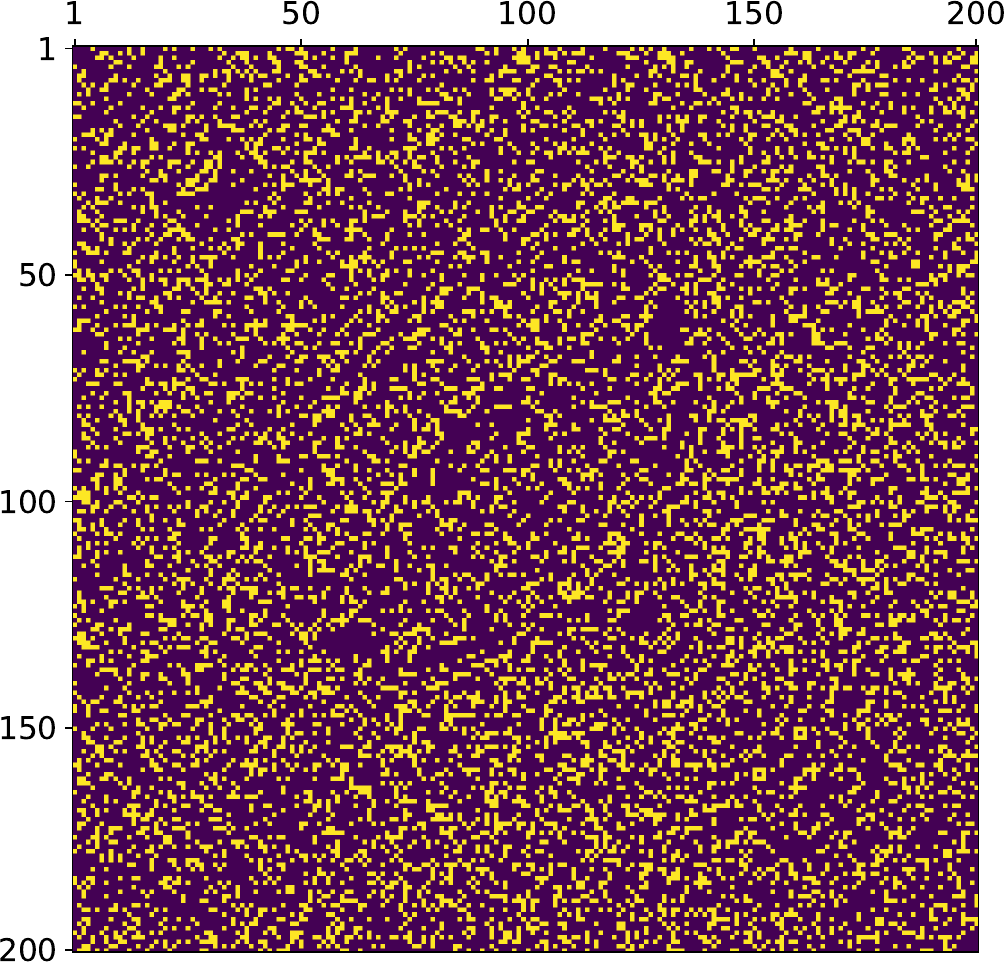}&
		\includegraphics[width=\FigWidth\linewidth]{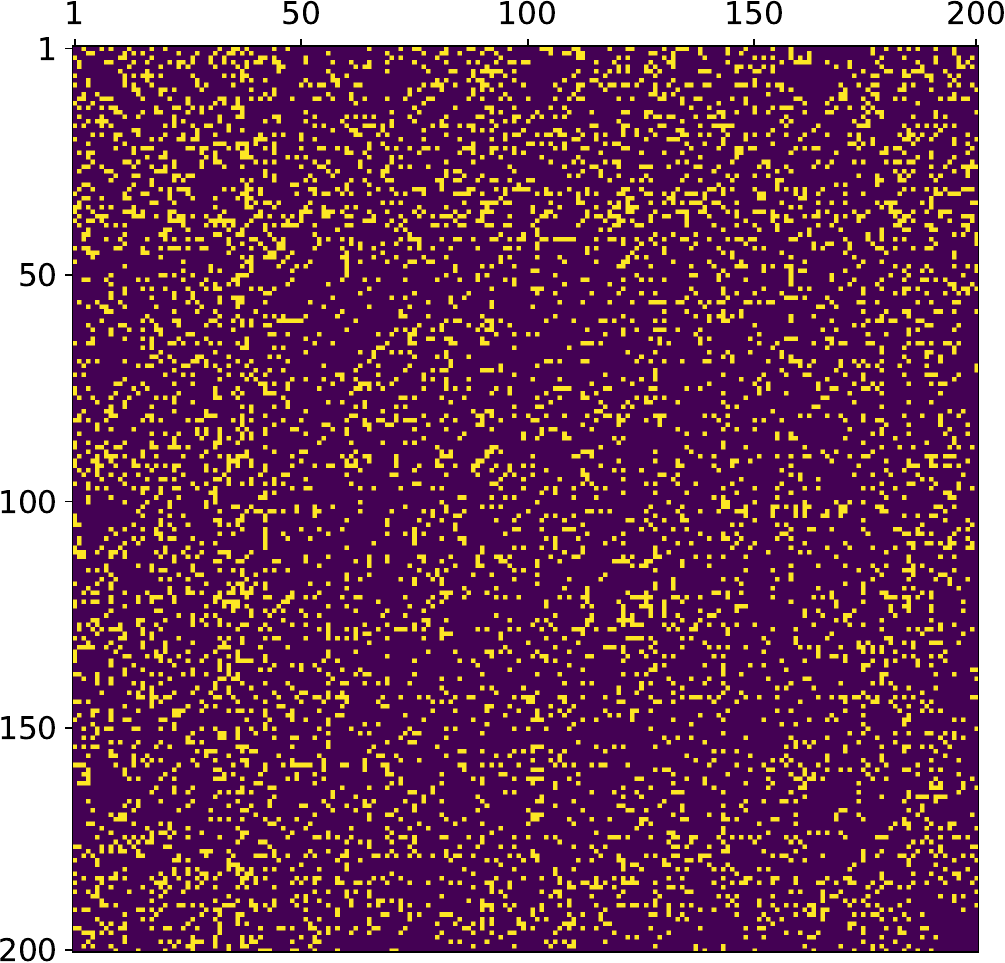}\\
		\rotatebox{90}{\hspace*{\spacing cm}$\zeta = 0.4$} & 
		\includegraphics[width=\FigWidth\linewidth]{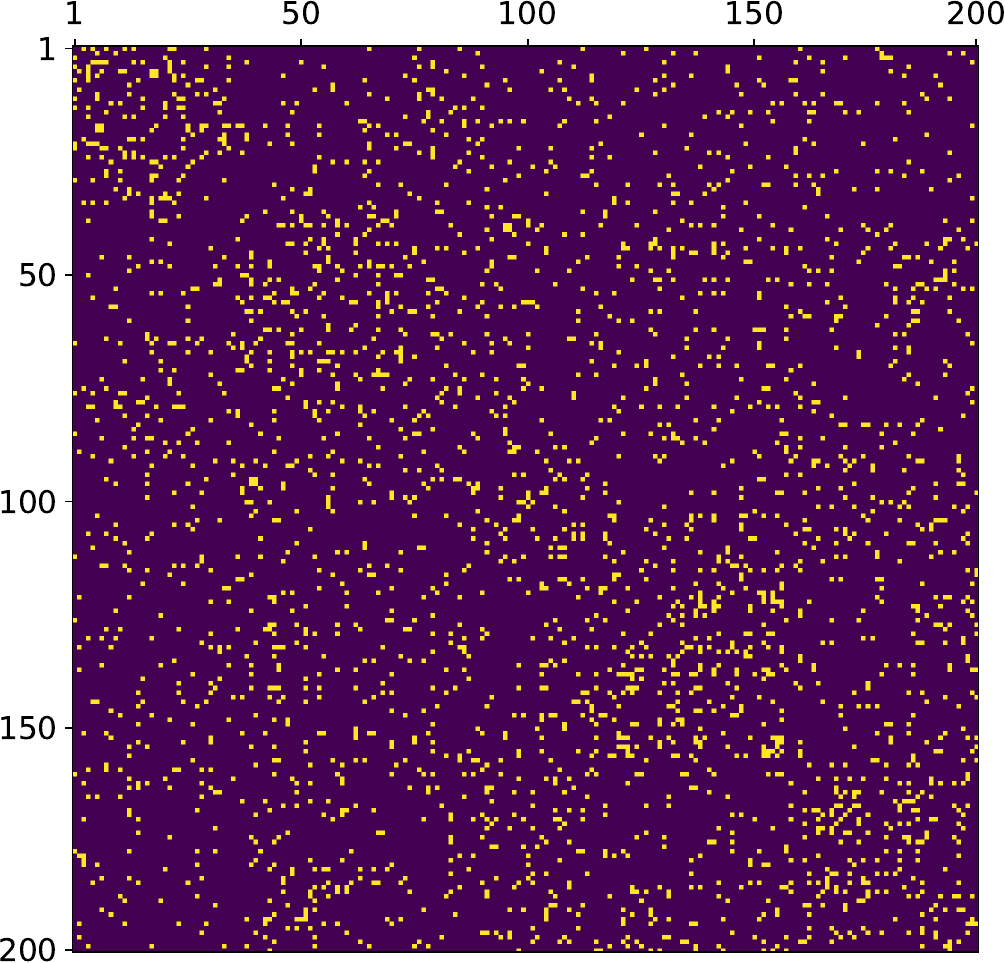}&
		\includegraphics[width=\FigWidth\linewidth]{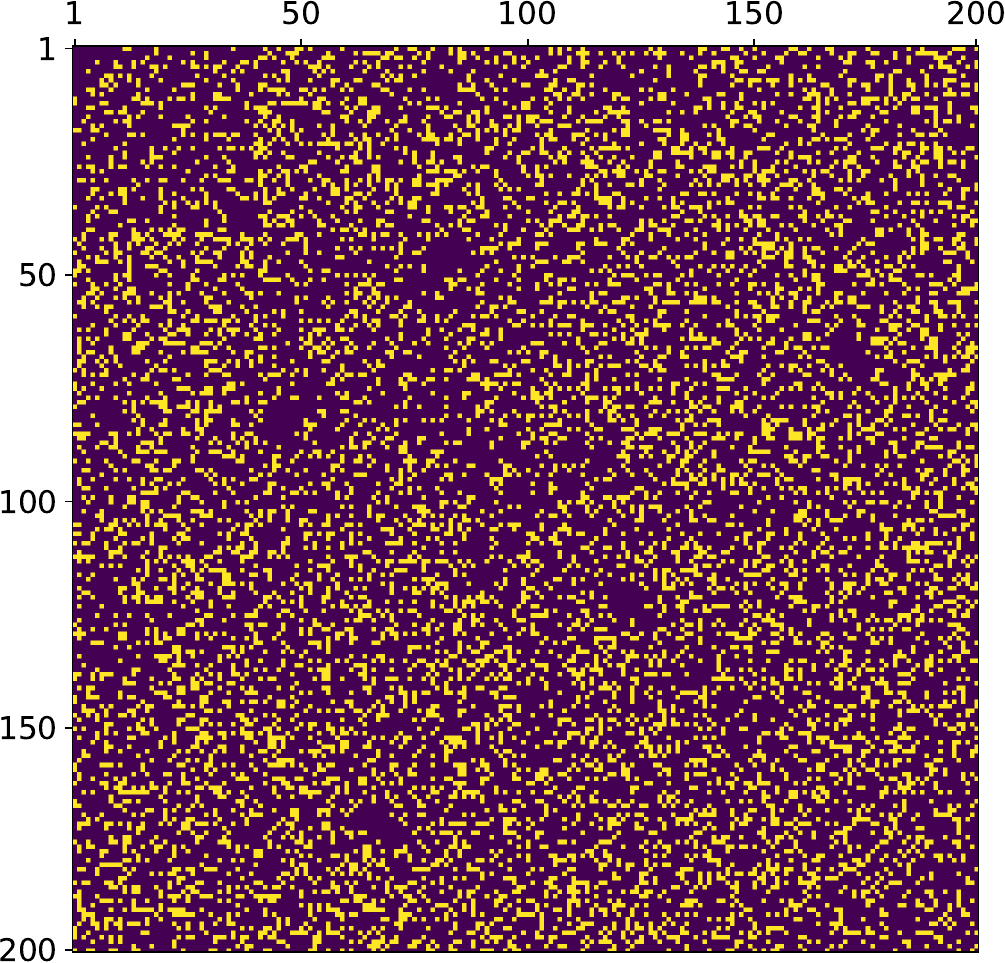}&
		\includegraphics[width=\FigWidth\linewidth]{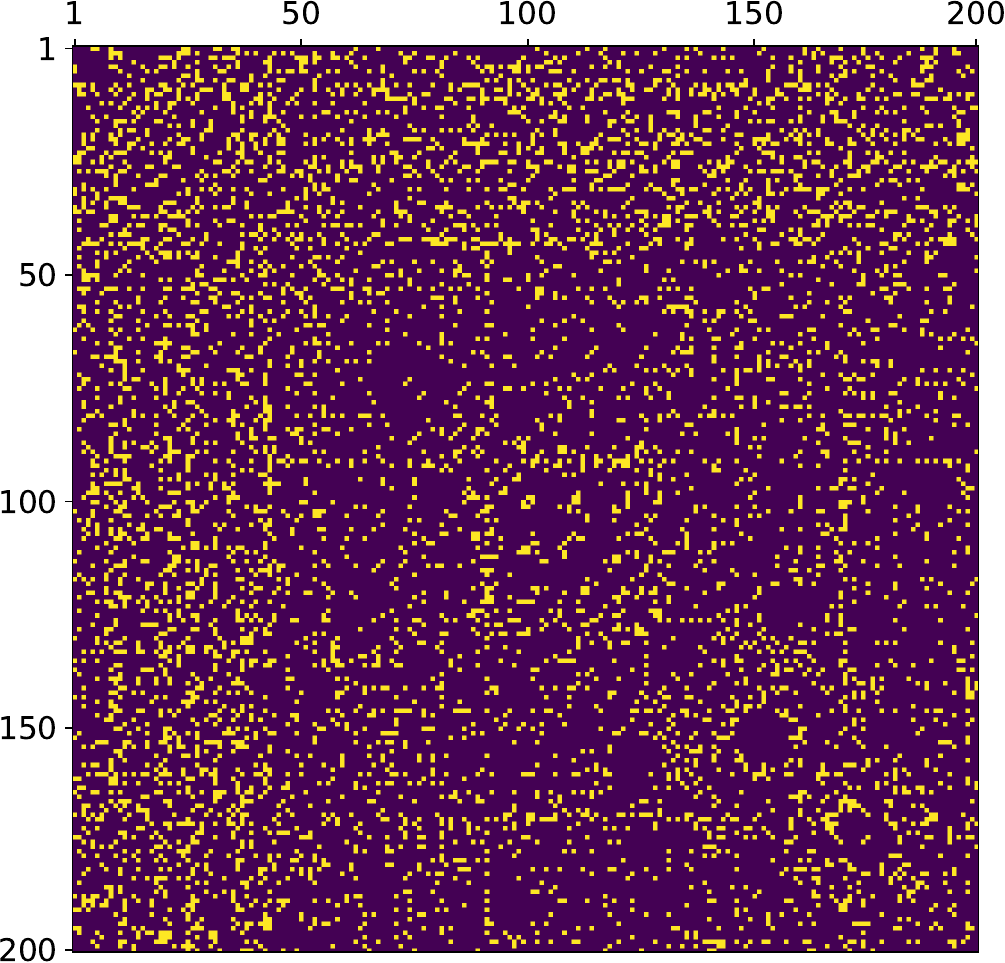}\\
		\rotatebox{90}{\hspace*{\spacing cm}$\zeta = 0.6$} & 
		\includegraphics[width=\FigWidth\linewidth]{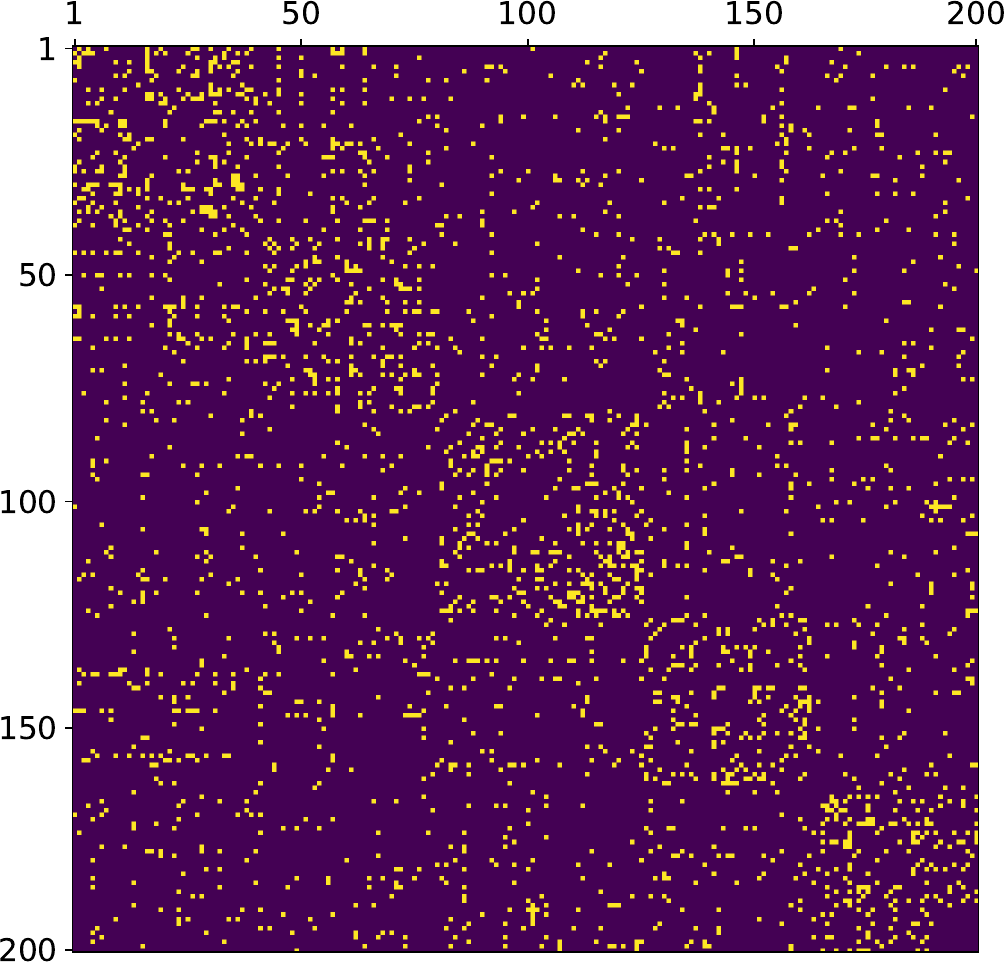}&
		\includegraphics[width=\FigWidth\linewidth]{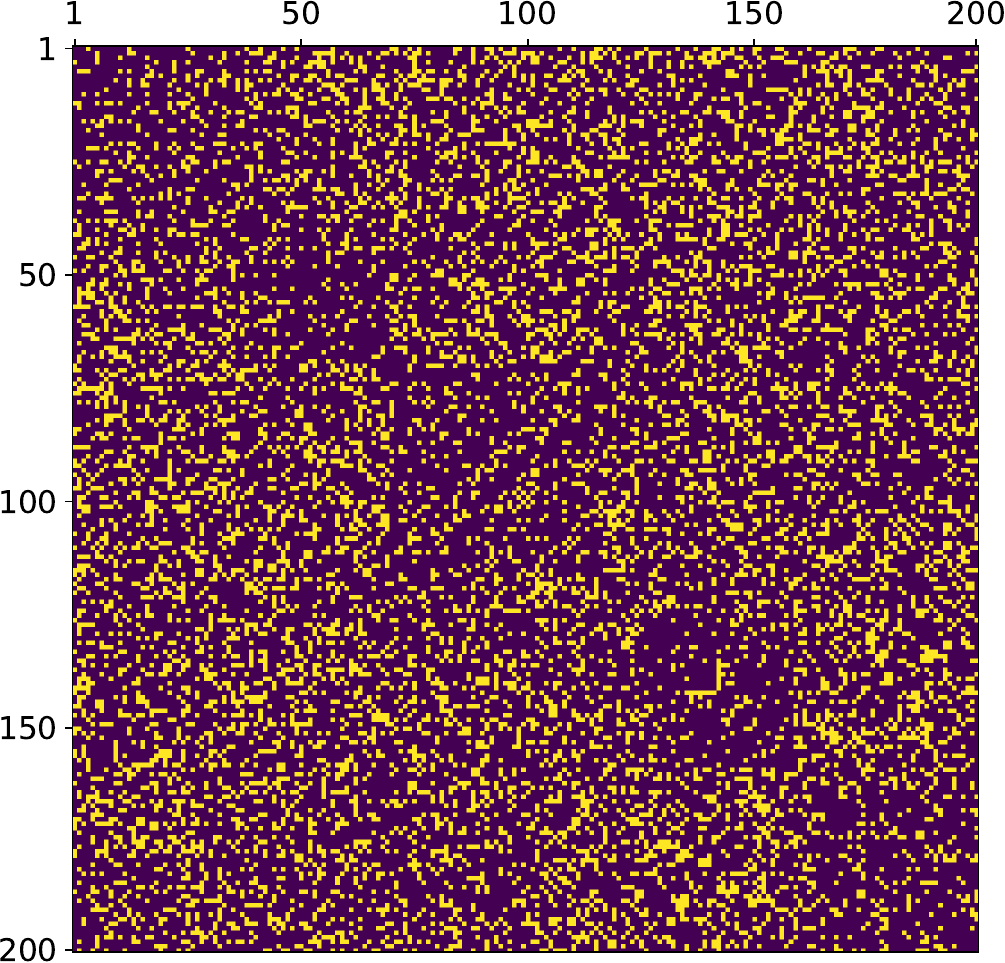}&
		\includegraphics[width=\FigWidth\linewidth]{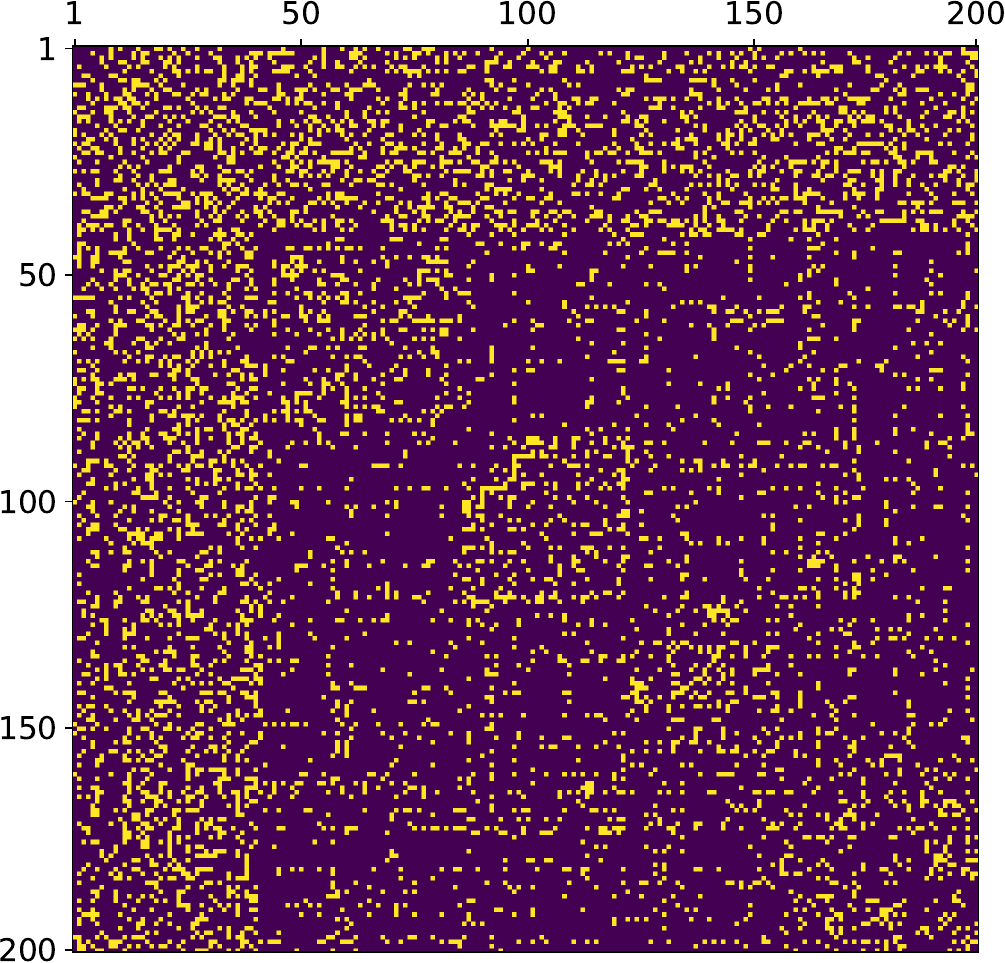}\\
		\rotatebox{90}{\hspace*{\spacing cm}$\zeta = 0.8$} & 
		\includegraphics[width=\FigWidth\linewidth]{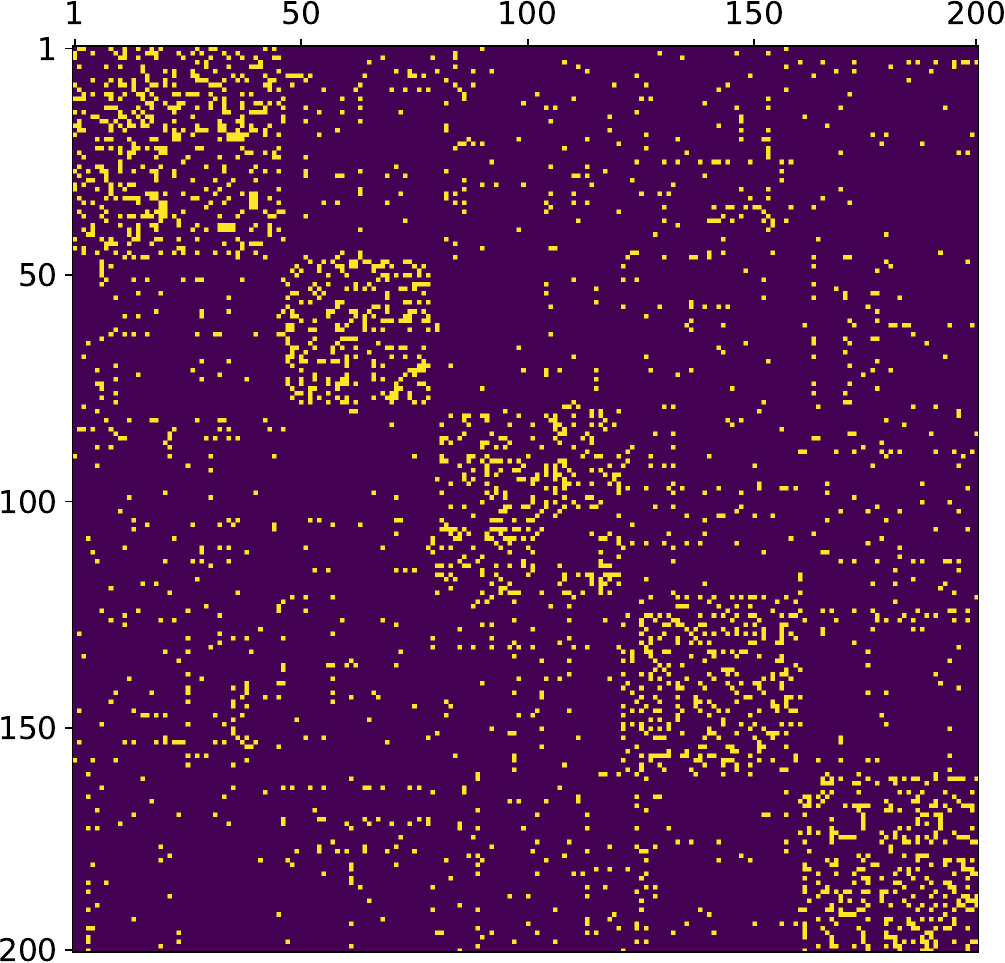}&
		\includegraphics[width=\FigWidth\linewidth]{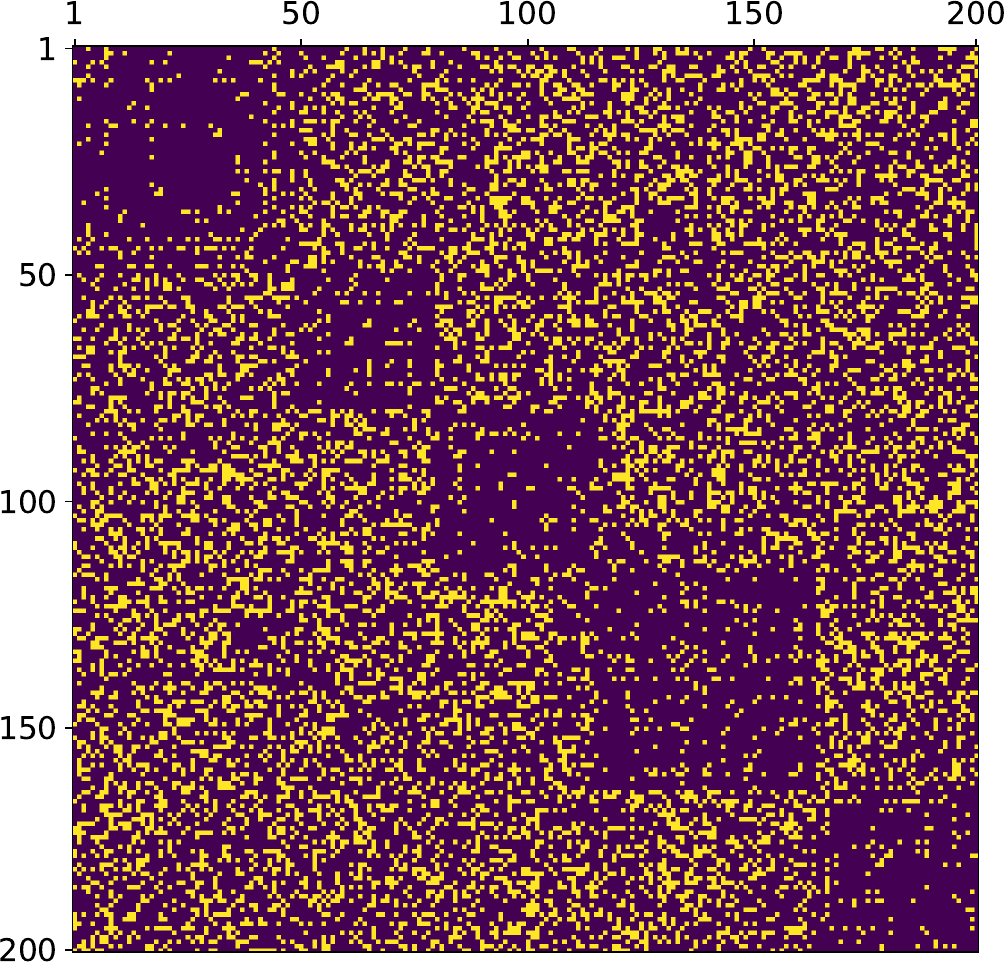}&
		\includegraphics[width=\FigWidth\linewidth]{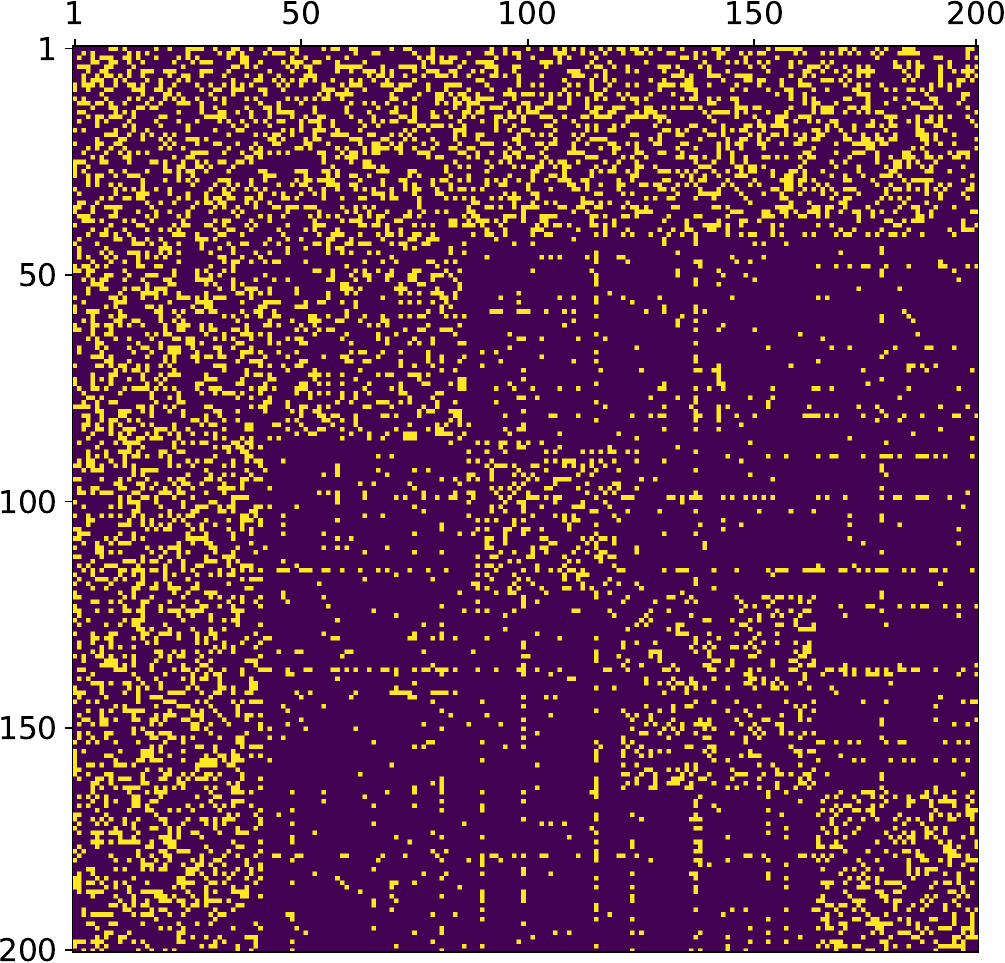}\\	
	\end{tabular}
	\captionof{figure}{Adjacency matrix sampled according to the partial membership setting, as described in Section 5 for different values of $\zeta$. A yellow (purple respectively) pixel at coordinate $(i,j)$ indicates the existence (the absence) of an edge between nodes $i$ and $j$ with $\beta$ equal to $0.3$.}
	\label{fig:appendix_adjacency_partial_memberships_details}
\end{table}

\subsection{Initialisation of the encoder}\label{sec:encoder_init}
VAEs are known to suffer from component collapsing. The generative model may enforce a regularisation preventing the model from learning any signal from the data. Strategies have been proposed to overcome this issue \parencite{higgins2017betavae}. Since the matrix $\bfPi$ has to be initialised as well as the encoder parameters, we propose another strategy consisting of obtaining a first estimation of communities in the network by running a K-Means algorithm \parencite{lloyd1982least} on the adjacency matrix. Given the obtained clusters, we run the encoder and minimise the mean-square error with the $\ell_2$-distance between  $(\mu_{\phi}(\bfA)_i)_i$ and the estimation obtained by the K-Means algorithm. We also aim at obtaining low variational variances $(\ln \sigma_{\phi}(\bfA)_i)_i$ by minimising the same mean-squared error between the $(\ln \sigma_{\phi}(\bfA)_i)_i$ and $0.01$. Other values have been tested in practice with negligible impact in all conducted experiments. The optimisation procedure as well as the initialisation is summarised in Algorithm 1.

\section{Synthetic data}

\subsection{Benchmark evaluating the clustering performances}
This section details the choice of the parameters concerning the competitors in the benchmark of the clustering performance. We evaluate the model against the adversarially regularised variational graph autoencoder \parencite[ARVGA,][]{pan2018adversarially}, with a $32$-dimensional hidden layer as well as for the latent space, and a $64$ dimensional hidden layer for the decoder. In addition, we provide the results of the variational graph autoencoder \parencite[VGAE,][]{kipf2016variational} with a $30$-dimensional hidden layer and a $4$-dimensional latent space, which resulted in better clustering than with higher dimensions. The results of the deep latent position cluster model \parencite[DLPM,][]{liang2022deep}, with a $64$-dimensional hidden layer and a $16$-dimensional latent space are also presented, as well as the results of the stochastic block model \parencite[SBM,][]{holland1983stochastic, daudin2008mixture}, with  random initialisation, denoted SBM random for short Table 1, and K-Means initialisation, denoted SBM kmeans in Table 1, and the variational Bayes latent position cluster model \parencite[VBLPCM,][]{salter2013variational}. The latent spaces dimensions were chosen as the ones providing the best results on the ARI task with $\beta$ set to $0.3$. Moreover, for the methodologies that do not perform node clustering, namely the VGAE and the ARGVA, a K-Means algorithm is fitted on the estimated posterior node embeddings with the true number of clusters. All methodologies are estimated using the true number of clusters, and the results are reported in Table 1. The best results is coloured in \red{red}, the second best in \blue{blue} and the third one in \green{green}. When two results are equal up to the standard deviation, they are identically coloured and if no signal is recovered, i.e with an ARI too low, no colouration is used.

\subsection{Evaluation as a partial memberships model}
\Cref{fig:partial_membership} is a larger version of the Figure 3 from the main paper.	
\begin{figure}
	\centering
	\includegraphics[width=0.7\linewidth]{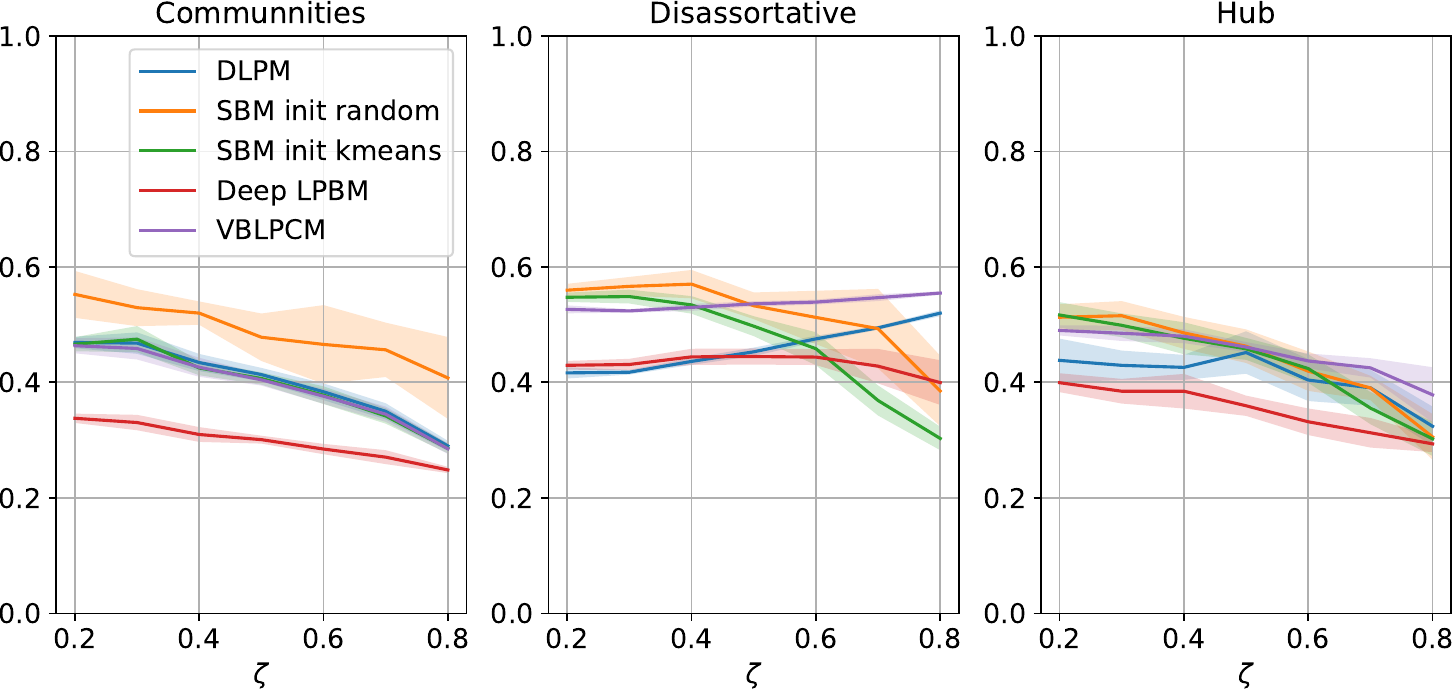}
	\caption{The partial memberships evaluation is obtained by plotting the results of the $\bfH$ metric, for different values of $\zeta$. The lower the quantity is, the better the estimation of $\bfeta$ is.}
	\label{fig:partial_membership}
\end{figure}

\section{French political blogosphere}

\subsection{Deep LPBM}
\Cref{fig:cropped-modelselectionaics} displays the different AIC value obtain for $Q$ varying from $2$ to $15$, with a highest value reached for $Q$ equal to $8$.
\begin{figure}
	\centering
	\includegraphics[width= 0.4\linewidth]{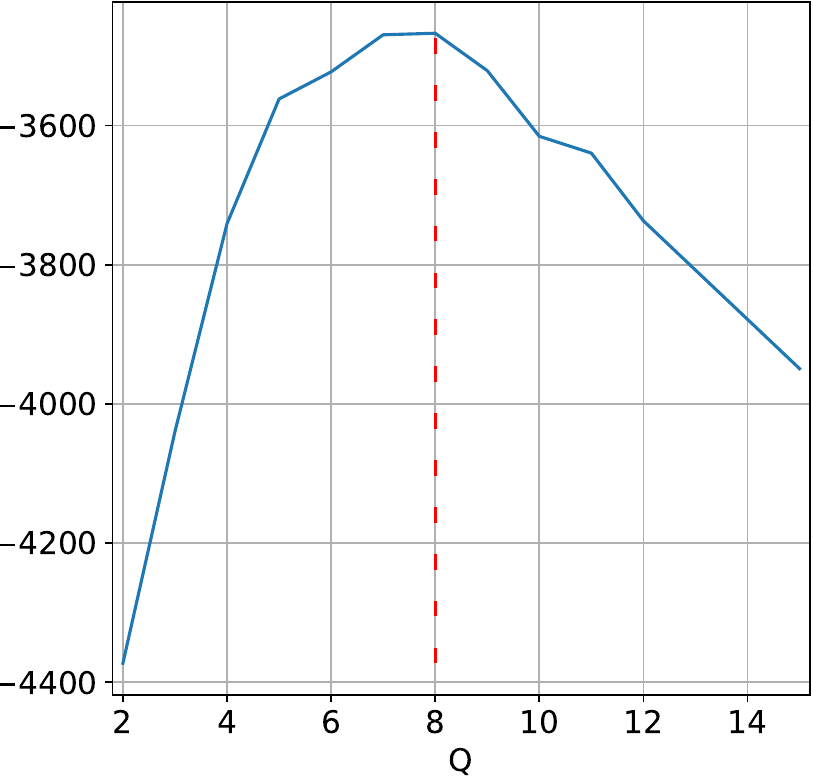}
	\caption{AIC values of Deep LPBM for Q varying from $2$ to $15$.}
	\label{fig:cropped-modelselectionaics}	
\end{figure}

Deep LPBM visualisation of the entier network is displayed in \Cref{fig:blog_network_visualisation_big}. In particular, the pie charts used for each node account for the corresponding node partial memberships.
\begin{figure}
	\centering
	\subcaptionbox{Political parties.\label{fig:blog_network_visualisation_political_parties}}[0.89\linewidth]{	\includegraphics[width=0.7\linewidth]{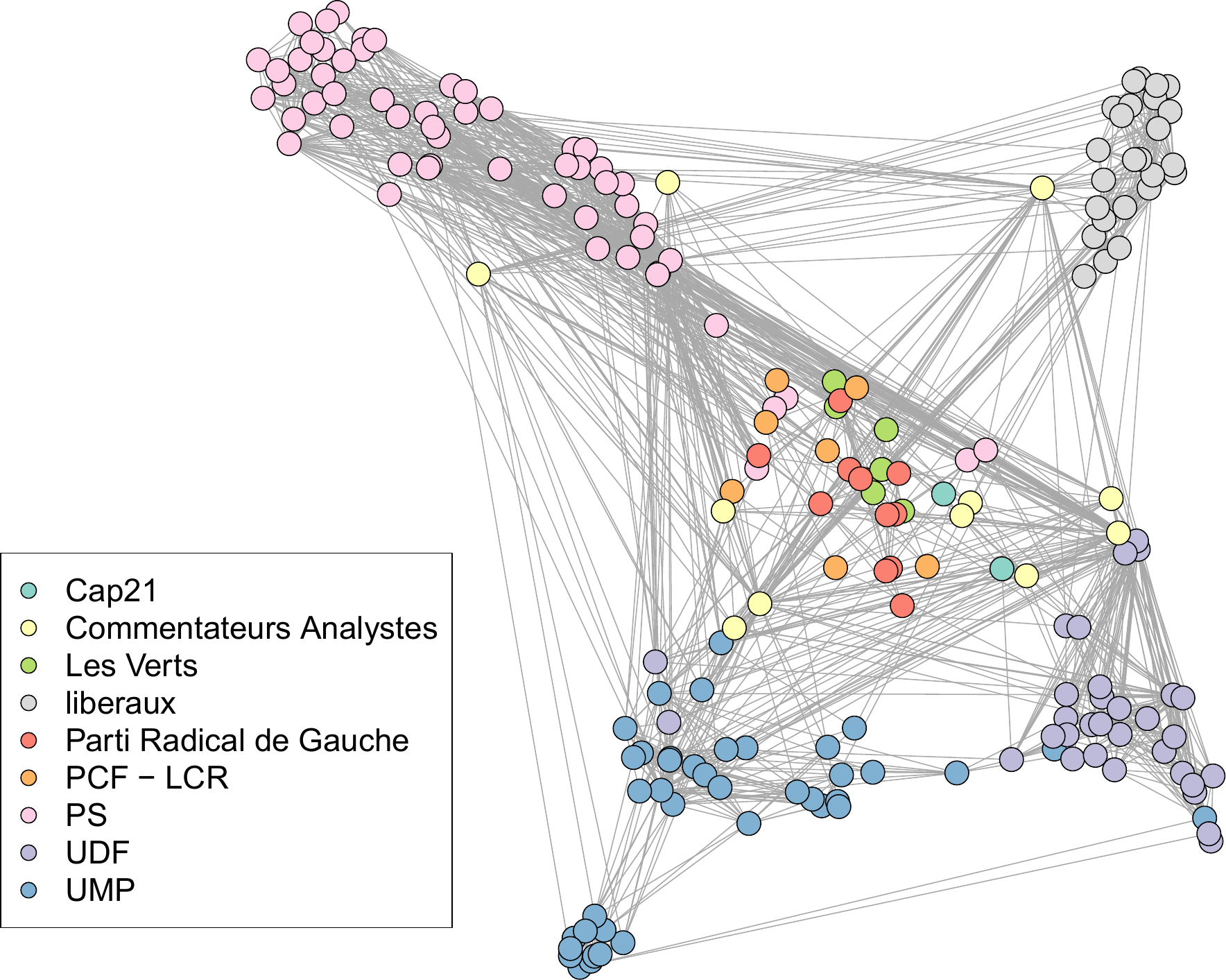}}
	\vfill
	\subcaptionbox{Deep LPBM partial memberships.		\label{fig:blog_network_visualisation_dlpbm}}[0.7\linewidth]{\includegraphics[width=\linewidth]{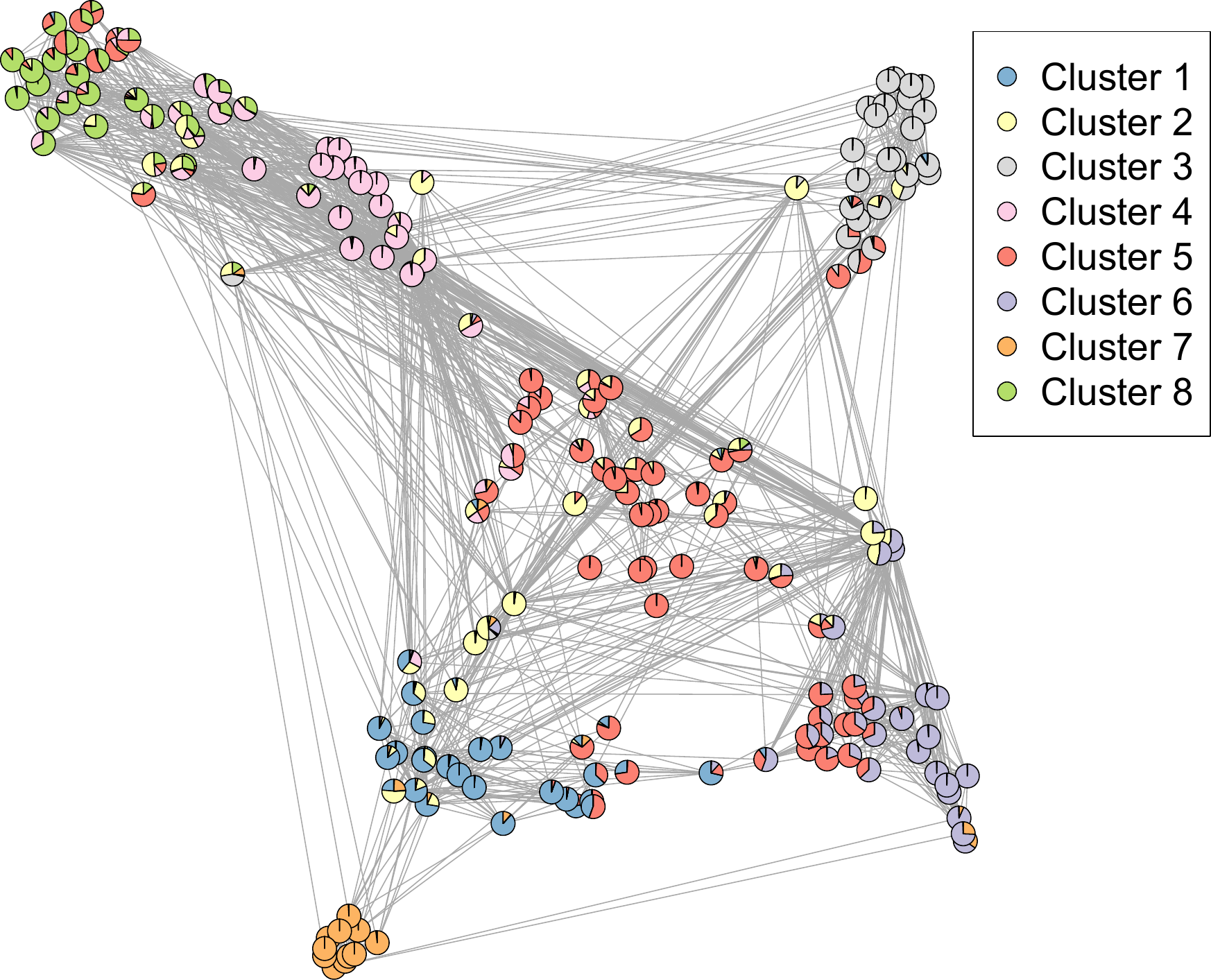}}
	\caption{All node positions were estimated by Deep LPBM. On the right-hand side, the node colours indicate the political party associated to the blog. On the left-hand side, each node is depicted by a pie chart representing the partial memberships estimated by Deep LPBM.}
	\label{fig:blog_network_visualisation_big}
\end{figure}

\subsection{Comparison with SBM results}\label{sec:real_data:sbm}

\begin{figure}
	\centering
	\begin{subfigure}{0.7\linewidth}
		\includegraphics[width=\linewidth]{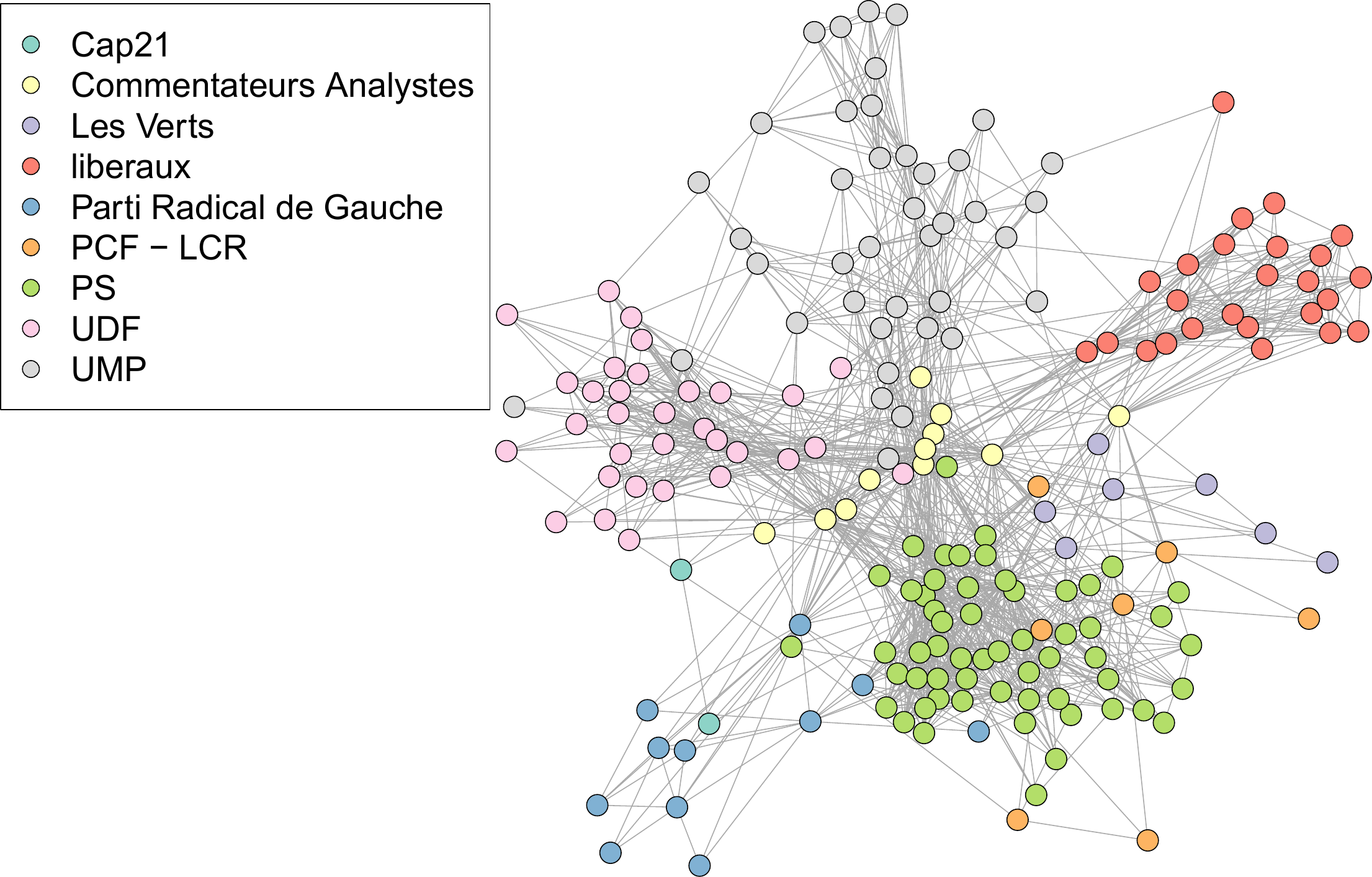}
		\caption{Political parties as node colours}
		\label{fig:blog_network_visualisation_true_parties}
	\end{subfigure}
	\vfill
	\begin{subfigure}{0.7\linewidth}
		\includegraphics[width=\linewidth]{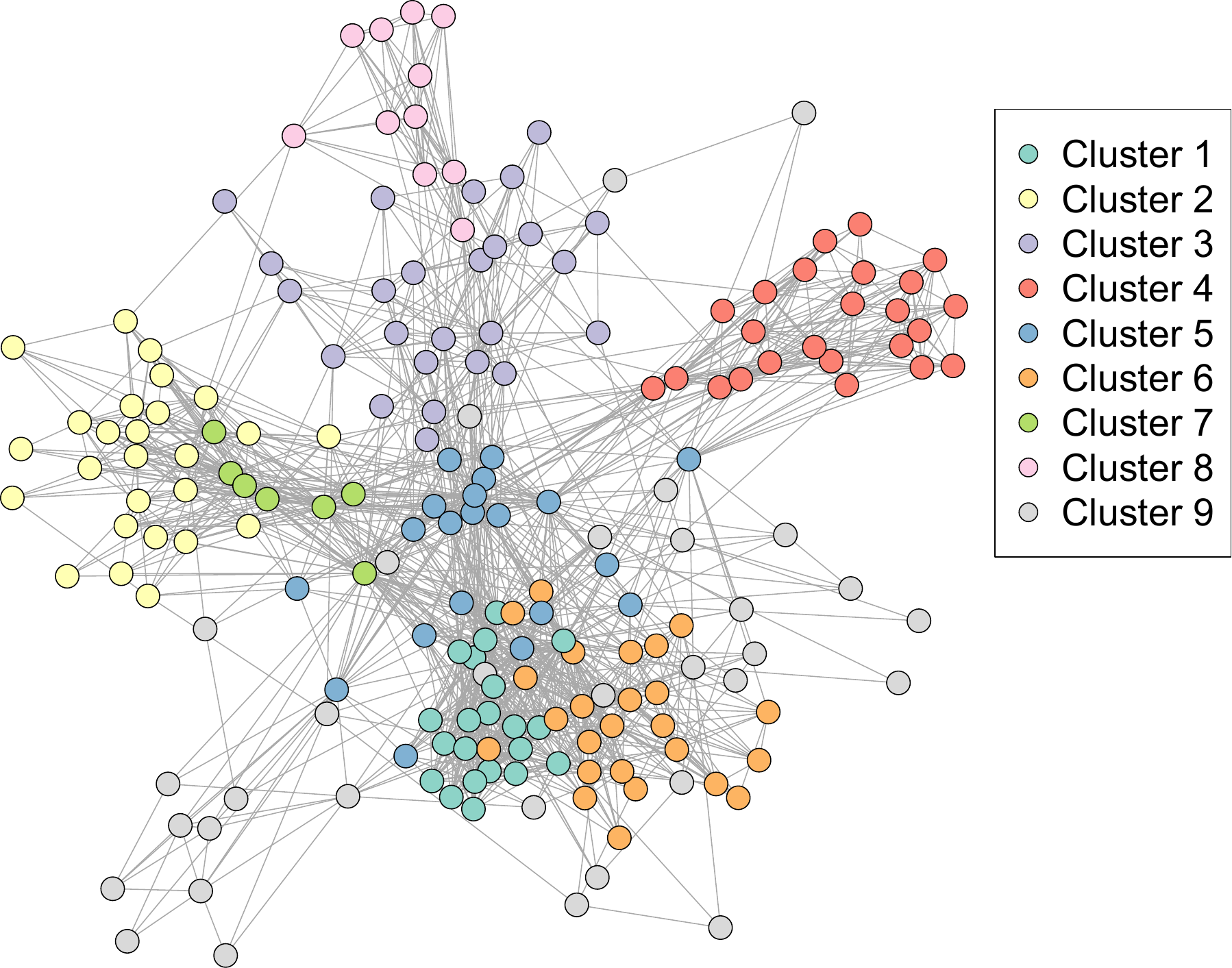}
		\caption{Clusters estimated by SBM.}
		\label{fig:blog_network_visualisation_SBM_clusters}		
	\end{subfigure}
	\caption{The node positions were computed using a Fruchterman Reingold algorithm \parencite{fruchterman1991graph}. On the left-hand side, the colour of the nodes corresponds to the political party the blog are associated with. On the right-hand side, the colour of the nodes indicate the SBM cluster assignments.}
	\label{fig:blog_network_visualisation_SBM}
\end{figure}
The section aims at presenting SBM results and comparing them to those obtained with Deep LPBM. In particular, we shall stress the differences due to the estimation of cluster memberships rather than partial memberships. As for Deep LPBM, we first estimated the best number of clusters using ICL for $Q$ varying from $2$ to $15$. The best number of clusters selected is $9$, as shown in  \Cref{fig:model_selection_sbm_blog} presented in the appendix.

Contrary to Deep LPBM, all SBM modelling assumptions are made at the cluster level. Hence, the connections between nodes are entirely deduced from their cluster membership assignments and the connectivity matrix $\bfPi$. Therefore, the matrix $\bfPi$ alone, provided in \Cref{fig:estimated_pi_matrix_sbm}, with the node cluster memberships suffices to determine the connectivity patterns beyond this graph generation. As for Deep LPBM, a ``garbage cluster'' emerges in the name of Cluster $9$, which regroups poorly connected nodes. Most of the other clusters have a probability of connection higher with nodes in the same cluster  than nodes from other clusters. However, two exceptions occur. First, Cluster $7$ is highly connected to Clusters $2$ and $5$. Second, Cluster $6$ has a high probability of connection to Cluster $1$, which corresponds to the two blog communities within the Socialist Party (PS).

A visualisation of the network is provided in \Cref{fig:blog_network_visualisation_SBM}. Since SBM does not provide a visualisation of the network, an external algorithm, namely the Fruchterman-Reingold algorithm \parencite{fruchterman1991graph}, had to be used. On the left-hand side, each node colour corresponds to the corresponding political party of the blog, while on the right-hand side, each node colour corresponds to SBM node cluster membership assignment. The first main difference with Deep LPBM is the limited compatibility between the node positions and the clustering. As an example, nodes in Cluster $9$ are spread across the entire network. Another salient difference backing the usefulness of estimating the node positions using the partial membership assignments is exposed by the UMP blogs. Indeed, SBM, as Deep LPBM, captures two strong communities within the UMP blogs and separates them into Clusters $3$ and $8$, as indicated by the $\bfPi$ matrix in \Cref{fig:estimated_pi_matrix_sbm}. However, the important discrepancy between the connectivity patterns of the two clusters is not translated on the visualisation in \Cref{fig:blog_network_visualisation_SBM_clusters} and is not made clear by the Fruchterman-Reingold algorithm, while captured by SBM estimation. To be fair, we also provide the visualisation obtained by the Fruchterman-Reingold algorithm with edge weights equal to the probability of connections of the corresponding node clusters. \Cref{fig:sbmvisualisationwithweigthsinkscape} in the appendix shows that using the edge weights can correct the effect mentioned above but at the cost of other issues, such as nodes collapsing on top of each other for instance.

To end this comparison between Deep LPBM and SBM results, we insist on similarities between the two results stressing the relevance of the discovered patterns composing the blog network. First, as stated above, UMP blogs, as well as PS blogs, are both separated into two communities as well Liberals (liberaux) blogs which are gathered into a single community. Cluster $2$ of Deep LPBM results and Cluster $5$ of SBM results correspond to nodes connected to many clusters, with a behaviour close to a hub. We can note that Deep LPBM results allow us to refine this observation by obtaining which nodes are connected to which clusters. Eventually, both methodologies use an extra cluster to regroup poorly connected nodes.

\paragraph{Fruchterman-Reingold with edge weights}
\begin{figure}[H]
	\centering
	\includegraphics[width=0.4\linewidth]{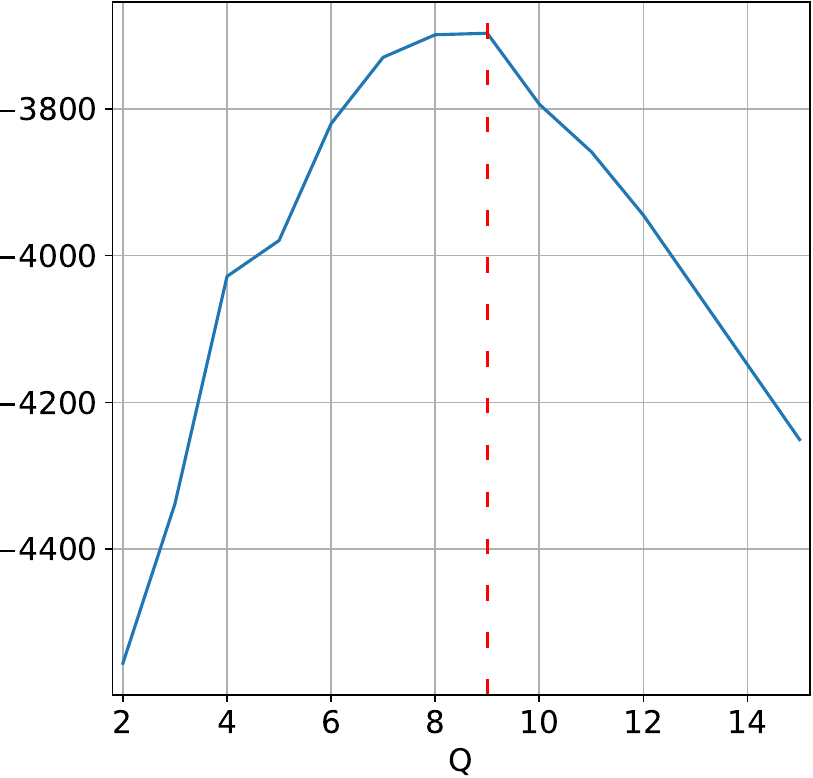}
	\caption{ICL of SBM for different number of clusters $Q$. The selected number of cluster is $9$.}
	\label{fig:model_selection_sbm_blog}
\end{figure}
In \Cref{fig:sbmvisualisationwithweigthsinkscape}, the node positions are obtained using a Fruchterman-Reingold algorithm with edge weights corresponding to the corresponding probability of connection between nodes. While this allows to obtain a better visualisation of the estimated communities, and in particular Clusters $3$ and $8$, it make the nodes in the middle indistinguishable. Indeed, to be able to obtain this figure, it was necessary to increase the size of the figure as well as decrease the node size, because nodes have collapsed into a single cluster. 

In addition, \Cref{fig:estimated_pi_matrix_sbm} provides the estimation of SBM $\bfPi$ matrix. This matrix shows the communities emerging the analysis, as well as the structure of blogs affiliated to the same political party. This is described in more detail in \Cref{sec:real_data:sbm}.
\begin{figure}
	\centering
	\includegraphics[width=0.4\linewidth]{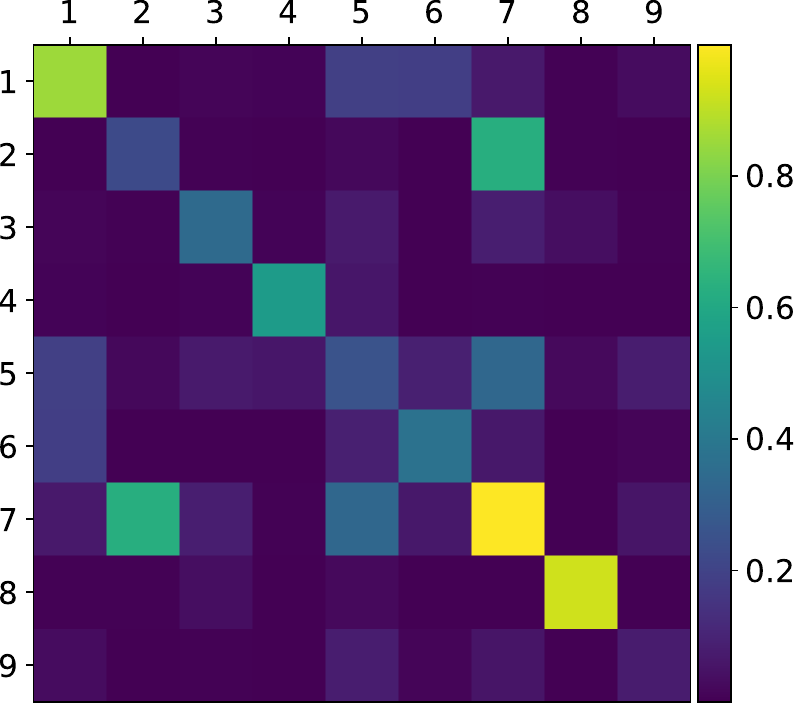}
	\caption{$\bfPi$ matrix estimated with SBM.}
	\label{fig:estimated_pi_matrix_sbm}
\end{figure}

The node positions in \Cref{fig:sbmvisualisationwithweigthsinkscape} were obtained by using the Fruchterman-Reingold algorithm with edge attribute equal to the term $\bfPi_{qr}$ with $q$ and $r$ the corresponding node cluster membership assignments. 
\begin{figure}
	\centering
	\includegraphics[width=0.7\linewidth]{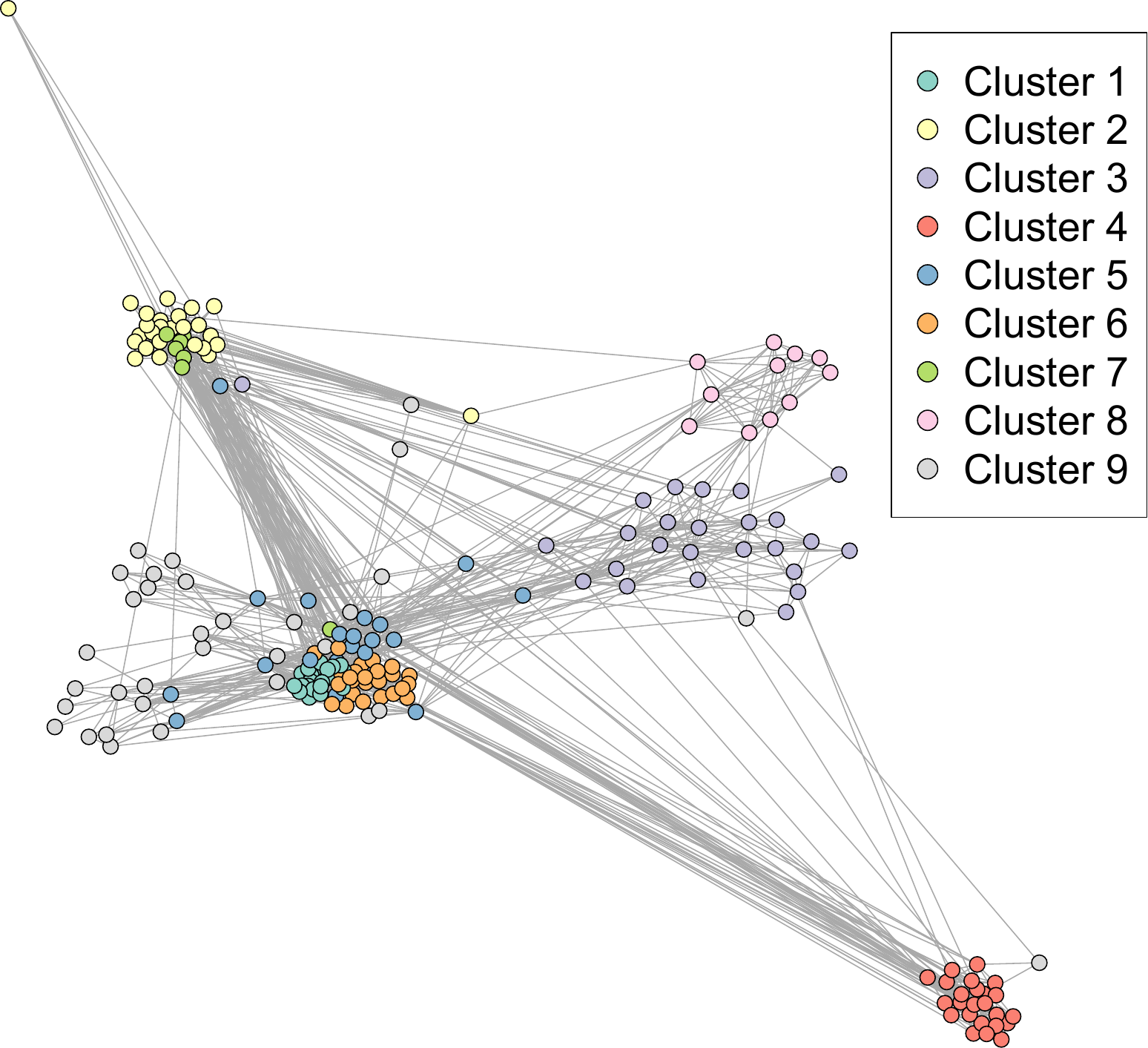}
	\caption{Visualisation of the network with Fruchterman-Reingold algorithm using the connectivity between groups as edge weights. The node colours correspond to the node cluster membership assignments.}
	\label{fig:sbmvisualisationwithweigthsinkscape}
\end{figure}

\end{document}

\typeout{get arXiv to do 4 passes: Label(s) may have changed. Rerun}